\RequirePackage[l2tabu, orthodox]{nag}
\documentclass[10pt]{article}

\usepackage[T1]{fontenc}
\usepackage[utf8]{inputenc}
\usepackage{lmodern}
\usepackage[authoryear]{natbib}

\usepackage{microtype}
\DeclareFontSeriesDefault{bf}{bx}          
\usepackage[frenchmath]{mathastext}		   
\usepackage{amsmath}                       
\numberwithin{equation}{section}
\usepackage{amsfonts}
\usepackage{graphicx}                      
\usepackage[plainpages=false, pdfpagelabels]{hyperref} 
	\hypersetup{
		colorlinks   = true,
		citecolor    = RoyalBlue, 
		linkcolor    = RubineRed, 
		urlcolor     = Turquoise
	}
\usepackage[dvipsnames]{xcolor}
\usepackage[all]{xy}
\usepackage{amssymb}                        
\usepackage[mathscr]{eucal}                 
\usepackage{dsfont}							
\usepackage[paperwidth=8.5in,paperheight=11.00in,top=1.25in, bottom=1.25in, left=1.00in, right=1.00in]{geometry}
\usepackage{mathtools}                      
\mathtoolsset{showonlyrefs=true}            
\usepackage{microtype}		                
\usepackage{bbm}                            
\usepackage{physics}                        
\usepackage{enumerate}                      
\usepackage{subcaption}                     
\usepackage{booktabs}                       
\usepackage{amsthm}                         
\allowdisplaybreaks                         

\theoremstyle{plain}
\newtheorem{theorem}{Theorem}
\numberwithin{theorem}{section}
\newtheorem{proposition}[theorem]{Proposition}
\newtheorem{corollary}[theorem]{Corollary}
\theoremstyle{definition}

\newtheorem{remark}[theorem]{Remark}

\newcommand{\keywords}[1]{%
  \vspace{1em}
  \noindent\textbf{Keywords: }#1
}

\newcommand{\pd}{\partial}
\newcommand{\one}{\mathbbm{1}}

\newcommand{\E}{\mathbb{E}}
\newcommand{\V}{\mathbb{V}}
\renewcommand{\S}{\mathbb{S}}
\newcommand{\K}{\mathbb{K}}

\newcommand{\Q}{\mathbb{Q}}

\newcommand{\R}{\mathbb{R}}

\newcommand{\cN}{\mathcal{N}}

\newcommand{\cL}{\mathcal{L}}
\newcommand{\cF}{\mathcal{F}}

\newcommand{\cA}{\mathcal{A}}

\newcommand{\cV}{\mathcal{V}}

\newcommand{\kD}{\mathfrak{D}}
\newcommand{\kK}{\mathfrak{K}}
\newcommand{\kT}{\mathfrak{T}}

\newcommand{\what}{\widehat}
\newcommand{\wtilde}{\widetilde}
\newcommand{\womega}{\what{\omega}}

\renewcommand{\(}{\left(}
\renewcommand{\)}{\right)}
\renewcommand{\[}{\left[}
\renewcommand{\]}{\right]}

\author{
    Jimin Lin\thanks{Quantitative Research, Bloomberg. (E-mail: \href{mailto:jlin846@bloomberg.net}{jlin846@bloomberg.net}).}
}

\begin{document}
\title{Shallow Representation of Option Implied Information}
\date{}
\maketitle

\begin{abstract}
Option prices encode the market's collective outlook through implied density and implied volatility. An explicit link between implied density and implied volatility translates the risk-neutrality of the former into conditions on the latter to rule out static arbitrage. Despite earlier recognition of their parity, the two had been studied in isolation for decades until the recent demand in implied volatility modeling rejuvenated such parity. This paper provides a systematic approach to build neural representations of option implied information. As a preliminary, we first revisit the explicit link between implied density and implied volatility through an alternative and minimalist lens, where implied volatility is viewed not as volatility but as a pointwise corrector mapping the Black-Scholes quasi-density into the implied risk-neutral density. Building on this perspective, we propose the neural representation that incorporates arbitrage constraints through the differentiable corrector. With an additive logistic model as the synthetic benchmark, extensive experiments reveal that deeper or wider network structures do not necessarily improve the model performance due to the nonlinearity of both arbitrage constraints and neural derivatives. By contrast, a shallow feedforward network with a single hidden layer and a specific activation effectively approximates implied density and implied volatility.

\end{abstract}

\section{Introduction}

Option prices implicitly contain the consensus of the market prospect on the underlying asset, and constructing a continuous representation of such implied information, mainly the \emph{implied density}, also known as the risk-neutral density, and the \emph{implied volatility (surface)}, from discrete observations of vanilla option prices has attracted interest for a long time. Implied density can be recovered by the second partial derivative of the option price \citep{breeden1978prices}. Implied volatility arose from the need to amend the option price produced by the Black-Scholes formula \citep{black1973pricing} to match the market price, and quickly became a central concept—along with sources of misconception—in modern option pricing theory. Implied density can be expressed as an explicit function of implied volatility, which essentially translates risk-neutrality into arbitrage conditions on implied volatility. We provide a version, Equation~\eqref{eq:psi_omega} in Proposition~\ref{prop:psi_omega}, which regards implied volatility as a numerical corrector that pointwise morphs the Black-Scholes quasi-density into implied density. We henceforth use the (implied-)density-(implied-)volatility parity as a pivotal idea to facilitate our discussion on how it relates to implied volatility modeling with modern machine learning facilities.

Curiously, the explicit density-volatility parity seems to be missing from the literature for a long period of time, and is then repeatedly rediscovered for different purposes in the early 2000s. \citep{jackwerth2000recovering} derived the parity to recover a risk aversion function in an economic context. \citep{brunner2003arbitrage} used the parity in the context of arbitrage analysis of implied volatility smiles. Part of the parity is also provided by \citep{durrleman2010implied} as the butterfly arbitrage condition, which might be traced back to an unpublished manuscript, as referenced in \citep{roper2009implied, roper2010arbitrage}. \citep{tavin2012implied} decomposed the formula as a lognormal density plus two adjustment terms to derive arbitrage conditions. Contextual segregation and simple derivation might cause the parity to be regarded as a trivial result, so this parity has never been given a formal name.

Sporadic attention to this parity reflects diverging interests among option researchers. The focus on implied density is likely driven by economic interest, where the risk premium is studied by contrasting empirical density and implied density without particular treatment on implied volatility. This interest is later absorbed by the notion of pricing kernel, see \citep{figlewski2018risk} for a review. The attention on implied volatility comes from two directions: option pricing modeling and implied volatility modeling. Modern option pricing modelers often use implied volatility as both a starting point and an end point to design and test their stochastic processes, such as \citep{bergomi2012stochastic, bayer2016pricing, jaber2022quintic}, where, in contrast to earlier generations of models like \citep{cox1976valuation, merton1976option, dupire1994pricing, heston1993closed, bates1996jumps}, attention to implied density is largely bypassed. On the other hand, attempts to directly model implied volatility will eventually return to the density-volatility parity for arbitrage conditions.

Early parameterizations of implied volatility surface, such as the stochastic volatility inspired (SVI) model \citep{gatheral2004parsimonious, gatheral2006volatility} and the stochastic alpha, beta, rho (SABR) model \citep{avellaneda2005sabr} did not enforce the density-volatility parity and were found to violate arbitrage constraints by \citep{roper2010arbitrage}. \citep{tavin2012implied} did not catch the violation in the SVI model. \citep{benko2007extracting} noticed the parity from \citep{brunner2003arbitrage} and employed it as a constraint for local polynomial smoothing of implied volatility. \citep{fengler2009arbitrage, andreasen2011volatility, glaser2012arbitrage} more or less acknowledged the parity, but did not incorporate it into their models directly. \citep{gatheral2014arbitrage} utilized the parity from \citep{roper2009implied} to modify the SVI model to rule out arbitrage, resulting in a class of models named surface SVI (SSVI). The parity dissolved into a specific condition for a SSVI parameter \citep[Theorem 4.2]{gatheral2014arbitrage}, and then became implicit in succeeding extended SSVI models \citep{hendriks2019extended, corbetta2019robust, mingone2022no}.

Closely following the implied volatility parameterization attempts were the usage of neural networks to represent implied volatility, where it is natural to adapt the parity and other arbitrage conditions expressed in terms of implied volatility, collected by \citep{roper2010arbitrage}, into loss functions. Deep smoothing came along this route first, named after \citep{ackerer2020deep}, which implemented feedforward neural networks to interpolate implied volatility surfaces. \citep{zheng2021incorporating} applied ensemble methods to build implied volatility as a weighted average from multiple neural networks. \citep{wiedemann2024operator} adopted a graph neural operator with smoothness regularization to jointly learn historical implied volatility surfaces, where multiple feedforward networks are used to approximate lifting, kernel convolution, and projection operations. \citep{yang2025hyperiv} chose a hypernetwork structure that used a transformer to encode the weight of another feedforward network that models implied volatility. There are generative approaches as well, such as \citep{bergeron2022variational, vuletic2024volgan} that use arbitrage conditions to screen out wrongly generated implied volatility surfaces.

Compared to the previous two decades' careful progression of parameterization approaches, the recent exploration of machine learning approaches seems radical. Although trending neural network structures are being tested, some basic questions of interest are still under discussion: Although there is a theoretical guarantee that a twice differentiable feedforward network can approximate implied volatility and implied density well \citep{hornik1989multilayer, hornik1990universal, hornik1991approximation}, yet how deep and wide a network we should use, and what activations are appropriate? Different underlying assets might have different scales. What is the optimal way to standardize the data? The market evolves and regimes change, and different option pricing models capture various market aspects, so what is a reasonable benchmark-and a potential prior?

In contrast to the aforementioned approaches that experiment with advanced neural networks, this paper is dedicated to building the simplest shallow neural representation of option implied information. We aim to provide a systematic treatment—from data processing, benchmark selection, and network design to experiment design and results analysis. The rest of the paper is structured as follows. In Section~\ref{sec:option}, we invite readers to review the option pricing theory without stochastic calculus through the lens of a minimalist. There are familiar changes of variables, whose economic or statistical meanings are often unappreciated, which naturally lead to expression simplification, dimension reduction, and data standardization. We also introduce the additive logistic pricing model \citep{carr2021additive, azzone2023fast, azzone2025explicit, lin2024neural} as the benchmark model. Section~\ref{sec:volatility} begins by clarifying why implied volatility is not volatility, exemplified by a family of risk-neutral logistic-beta distributions of a fixed volatility that produces rich shapes of implied volatility curves. Treating implied volatility as a pointwise density corrector intuitively leads to Proposition~\ref{prop:psi_omega}, which not only translates risk-neutrality into arbitrage conditions, but also justifies the usage of a feedforward neural network to approximate the corrector. Section~\ref{sec:neural} establishes the neural representation framework and shows that deep smoothing actually does not require the use of deep neural networks. Through comprehensive experiments and diagnosis, we identify a modeling risk that goodness of fit to implied volatility does not guarantee a correct implied density. We find that upscaling neural networks does not necessarily improve performance of neural representations due to the nonlinearity of arbitrage conditions and neural derivatives. Finally, we show by experiment that neural representation can be built with a feedforward network with only one hidden layer.

\section{Option Pricing Simplified} \label{sec:option}
For simplicity, we consider assets with positive prices. For treatment of negative prices, see \citep{choi2022black}. Let $S_t$ denote the price of the underlying asset at $t$. We assume the instantaneous risk-free interest rate $r_t$ and the continuous dividend rate $d_t$ are deterministic and positive. The forward price of the asset from $t$ to some future time $T$ is given by $F_{t,T} = S_t e^{\int_t^T (r_u - d_u) \dd u}$. A European vanilla put (resp. call) option with expiry $T$ and strike price $K$ offers the right to sell (resp. buy) the asset with price $K$ at time $T$, which gives a payoff of $(K - S_T)^+$ (resp. $(S_T - K)^+$). 

The market price of vanilla European put option, $P_t(T, K)$, and call option, $C_t(T, K)$, are given by the discounted expected payoff under the market risk-neutral measure $\Q$:
\begin{align} \label{eq:EPC}
P_t(T, K) = e^{-\int_t^T r_u \dd u}\E^\Q \[(K - S_T)^+ | S_t\],
&&
C_t(T, K) = e^{-\int_t^T r_u \dd u}\E^\Q \[(S_T - K)^+ | S_t\].
\end{align}

\subsubsection*{Innate dimension reduction and data standardization}
The following changes of variables translate the above prices expressed in absolute calendar time and dollar values into relative and dimensionless quantities. Define the tenor $\tau:= T - t$ and set $F_\tau := F_{t, T}$. Define
\begin{align} \label{eq:xk}
X_\tau := \log \frac{S_T}{F_\tau},
&&
\kappa := \log \frac{K}{F_\tau},
\end{align}
where $X_0 = 0$. They have intuitive economic interpretations. $X_\tau$ is the \emph{dividend adjusted logarithmic excess return} (\emph{return} for short). $\kappa$ is the \emph{required excess return to hit the strike}, which is commonly called the log-forward-moneyness (\emph{moneyness} for short).

Options quoted by expiry and strike are readily expressed in tenor and moneyness by $(T, K) = (\tau + t, F_\tau e^\kappa)$. Dividing Equation~\eqref{eq:EPC} by dividend adjusted spot price $S_t e^{-\int_t^\tau d_u \dd u}$, we define
\begin{align} \label{eq:Epc}
p(\tau, \kappa):= \frac{P_t(\tau + t, F_\tau e^\kappa)}{S_t e^{-\int_t^\tau d_u \dd u}} = \E^\Q \[\(e^\kappa - e^{X_\tau}\)^+\],
&&
c(\tau, \kappa):= \frac{C_t(\tau + t, F_\tau e^\kappa)}{S_t e^{-\int_t^\tau d_u \dd u}} = \E^\Q \[\(e^{X_\tau} - e^\kappa\)^+\].
\end{align}
$p$ is the \emph{put-to-(dividend-adjusted-)spot price ratio} and $c$ is the \emph{call-to-(dividend-adjusted-)spot price ratio}. For convenience, we still refer to them by put price and call price.

This transformation abstracts away $r$, $d$, and $S_t$ and aligns $X_\tau$, $\kappa$, $p$, and $c$ on the same scale. Thus, it yields inherent dimension reduction and data standardization.

\subsubsection*{Generic form of pricing formula}
Let $\psi(\tau, \cdot)$ and $\Psi(\tau, \cdot)$ denote the probability density function (PDF) and cumulative distribution function (CDF, or distribution) of $X_\tau$ under the $\Q$-measure. Note that $X$ is an exponential $\Q$-martingale, so it holds that $\E^\Q[e^{X_\tau}] = \int e^x \psi(\tau, x) \dd x = 1$, which implies that $\wtilde{\psi}(\tau, x) := e^x \psi(\tau, x)$ is already an Esscher exponential tilting of $\psi$. In other words, $\wtilde{\psi}$ is also a PDF. Let $\wtilde{\Psi}$ be the corresponding CDF of $\wtilde{\psi}$. Also define the complement of the CDF by $\overline{\Psi} = 1 -{\Psi}$. Then Equation~\eqref{eq:Epc} can be simplified to the following form:
\begin{align} \label{eq:pc}
p(\tau, \kappa) = e^\kappa \Psi(\tau, \kappa) - \wtilde{\Psi}(\tau, \kappa),
&&
c(\tau, \kappa) = \overline{\wtilde{\Psi}}(\tau, \kappa) - e^\kappa \overline{\Psi}(\tau, \kappa).
\end{align}
Equation~\eqref{eq:pc} explicates the relation between option price and the implied CDF, and also indicates that the essence of an option pricing model is to specify the risk-neutral CDF.

\subsubsection*{Sufficient conditions on price to exclude static arbitrage}
\cite{carr2005note} lists sufficient conditions on option prices to exclude static arbitrages. Here we express them in terms of the put option:
\begin{enumerate}[i.]
    \item Calendar spread: for any $\Delta \tau > 0$, $p(\tau + \Delta \tau, \kappa) \ge p(\tau, \kappa)$ due to convexity of vanilla payoff;
    \item Vertical spread: for any $\Delta \kappa > 0$, $p(\tau, \kappa + \Delta \kappa) - p(\tau, \kappa) \ge 0$ since the first payoff dominates the second;
    \item Butterfly spread: $p(\tau, \kappa + \Delta \kappa) - 2 p(\tau, \kappa) + p(\tau, \kappa - \Delta \kappa) \ge 0$ as the portfolio has possible positive payoff.
\end{enumerate}
These conditions can be expressed with partial derivatives as follows:
\begin{align} \label{eq:con_arb_p}
\pd_\tau p(\tau, \kappa) \ge 0,
&&
\pd_\kappa p(\tau, \kappa) \ge 0,
&&
\pd_{\kappa \kappa} p(\tau, \kappa) \ge 0.
\end{align}
Calculating these partial derivatives from Equation~\eqref{eq:pc}, we obtain
\begin{equation} \label{eq:pd_p}
\begin{aligned}
\pd_\tau p(\tau, \kappa)
    &= e^\kappa \pd_\tau \Psi(\tau, \kappa) - \pd_\tau \wtilde{\Psi}(\tau, \kappa),
\\
\pd_\kappa p(\tau, \kappa)
    &= e^\kappa \Psi(\tau, \kappa),
\\
\pd_{\kappa \kappa} p(\tau, \kappa)
    &= e^\kappa \Psi(\tau, \kappa) + e^\kappa \psi(\tau, \kappa).
\end{aligned}
\end{equation}
Notice that the positive vertical spread condition $\pd_\kappa p \ge 0$ and positive butterfly spread condition $\pd_{\kappa \kappa} p \ge 0$ inherit from the risk-neutrality of the distribution. Only the last positive calendar spread condition $\pd_\tau p \ge 0$ requires a specific term structure of the distribution.

\subsubsection*{Black-Scholes model}
 In the Black-Scholes model, the marginal distribution of the underlying return follows a risk-neutral normal distribution, whose variance increases linearly with tenor and whose mean equals negative one-half the variance, $X_\tau \sim N(- \frac{1}{2} \sigma^2 \tau, \sigma^2 \tau)$. The volatility $\sigma > 0$ is the only free parameter. Defining the total volatility $\omega$ absorbs the term structure and simplifies expressions:
\begin{equation} \label{eq:omega_sigma_const}
\omega(\tau):= \sigma \sqrt{\tau}
\end{equation}
The corresponding Black-Scholes PDF, CDF, tilted PDF, and tilted CDF are
\begin{align} \label{eq:pdf_n}
\psi_{BS}(\tau, x)
    &= \frac{1}{\omega(\tau)} \phi\big(z_+(\tau, x)\big),
&
\wtilde{\psi}_{BS}(\tau, x)
    &= \frac{1}{\omega(\tau)} \phi\big(z_-(\tau, x)\big),
\\
\Psi_{BS}(\tau, x)
    &= \Phi \big( z_+(\tau, x) \big),
&
\wtilde{\Psi}_{BS}(\tau, x)
    &= \Phi \big( z_-(\tau, x) \big),
\end{align}
where $\phi$ and $\Phi$ are respectively the PDF and CDF of the standard normal distribution $N(0,1)$ and 
\begin{equation} \label{eq:z_n}
z_\pm(\tau, x) := \frac{x \pm \frac{1}{2} \omega^2(\tau) }{\omega(\tau)} 
\end{equation}
are statistical pivots of $N(\mp \frac{1}{2} \omega^2(\tau), \omega^2(\tau))$. Here $\omega^2(\tau) = (\omega(\tau))^2$ is a shorthand.

The Black-Scholes pricing formula is readily obtained by plugging Equation~\eqref{eq:pdf_n} into Equation~\eqref{eq:Epc} and using the symmetry of normal distribution:
\begin{align} \label{eq:pc_n}
p_{BS}(\tau, \kappa)
    &= e^\kappa \Psi_{BS}(\tau, \kappa) - \wtilde{\Psi}_{BS}(\tau, \kappa)
    = e^\kappa \Phi\big(z_+(\tau, \kappa)\big) - \Phi\big(z_-(\tau, \kappa)\big),
\\
c_{BS}(\tau, \kappa)
    &= \overline{\wtilde{\Psi}}_{BS}(\tau, \kappa) - e^\kappa \overline{\Psi}_{BS}(\tau, \kappa)
    =\Phi \big( -z_-(\tau, \kappa) \big) - e^\kappa \Phi \big( -z_+(\tau, \kappa) \big).
\end{align}
Note that the pivots $z_-$ and $z_+$ are rewritten forms of classical Black-Scholes terms $-d_1$ and $-d_2$ to make their financial and statistical interpretations explicit: $z_+$ (resp. $z_-$) is the location-scale standardized return under the Black-Scholes PDF (resp. tilted PDF).

\subsubsection*{Additive logistic model as the benchmark model}
We will use the additive logistic model as our benchmark model, which was introduced by \citep{carr2021additive} and extensively studied by \citep{azzone2023fast, azzone2025explicit, lin2024neural}. Distinguished from models built bottom-up from stochastic differential equations, the additive logistic model provides direct top-down access to the risk-neutral density of the underlying asset return. For every tenor, three parameters can be specified to control the dispersion, right tail, and left tail of the marginal distribution, which is logistic-beta \citep{balakrishnan1991handbook}. The term structure is customized by the choice of interpolating every parameter into a continuous function of tenor, and by such, the underlying martingale process is easily constructed. It has an explicit pricing formula and is flexible to fit the real data in usual market conditions. Those features make the additive logistic model an ideal benchmark for new model development.

In the additive logistic model, every marginal of the return follows a risk-neutral logistic-beta distribution, $X_\tau \sim LB(\mu(\tau), \varsigma(\tau), \alpha(\tau), \beta(\tau))$. The location parameter is fixed due to martingality as 
\begin{align} \label{eq:lb_mu}
\mu(\tau) = \log\left(\frac{B(\alpha(\tau) + \varsigma(\tau), \beta(\tau) - \varsigma(\tau))}{B(\alpha(\tau), \beta(\tau))}\right),
\end{align}
where $B$ is the beta function. The three free parameters $\varsigma$, $\alpha$, and $\beta$ are functions of $\tau$ that respectively determine the term structure of dispersion, right tail, and left tail. The PDF $\phi_{LB}$ and CDF $\Phi_{LB}$ of standard logistic-beta distribution $LB(0, 1, \alpha, \beta)$ are given by
\begin{equation} \label{eq:pdf_lb}
\begin{aligned}
\phi_{LB}(z; \alpha, \beta)
    &= \frac{1}{B(\alpha, \beta)} \phi_L(z) (\Phi_L(z))^{\alpha - 1} (1 - \Phi_L(z))^{\beta - 1},
\\
\Phi_{LB}(z; \alpha, \beta)
    &= \int_0^{\Phi_L (z)} u^{\alpha - 1} (1-u)^{\beta - 1} \dd u,
\end{aligned}
\end{equation}
where $\phi_L$ and $\Phi_L$ are the PDF and CDF of standard logistic distribution $L(0,1)$:
\begin{align} \label{eq:pdf_l}
\phi_L(z) = \frac{e^z}{(1 + e^z)^2},
&&
\Phi_L(z) = \frac{e^z}{1 + e^z}.
\end{align}
Define $z(\tau, x):= \frac{x - \mu(\tau)}{\varsigma(\tau)}$ as the pivot. The risk-neutral PDF, tilted PDF, CDF, and tilted CDF of $X_\tau$ are
\begin{align}
\psi_{LB}(\tau, x)
    &= \frac{1}{\varsigma(\tau)} \phi_{LB}\big(z(\tau, x); \alpha(\tau), \beta(\tau)\big),
&
\wtilde{\psi}_{LB}(\tau, x)
    &= \frac{1}{\varsigma(\tau)} \phi_{LB}\big(z(\tau, x); \alpha(\tau) + \varsigma(\tau), \beta(\tau) - \varsigma(\tau)\big),
\\
\Psi_{LB}(\tau, x)
    &= \Phi_{LB}\big( z(\tau, x); \alpha(\tau), \beta(\tau) \big),
&
\wtilde{\Psi}_{LB}(\tau, x)
    &= \Phi_{LB}\big( z(\tau, x); \alpha(\tau) + \varsigma(\tau), \beta(\tau) - \varsigma(\tau) \big),
\end{align}
Plugging Equation~\eqref{eq:pdf_lb} into Equation~\eqref{eq:Epc} and using the symmetry of the logistic-beta distribution yields the option pricing formula:
\begin{equation} \label{eq:pc_lb}
\begin{aligned} 
p_{LB}(\tau, \kappa)
    &= e^\kappa \Psi_{LB}(\tau, \kappa) - \wtilde{\Psi}_{LB}(\tau, \kappa) \\
    &= e^\kappa \Phi_{LB}\big(z(\tau, \kappa); \alpha(\tau), \beta(\tau) \big) - \Phi_{LB}\big(z(\tau, \kappa); \alpha(\tau) + \varsigma(\tau), \beta(\tau) - \varsigma(\tau)\big),
\\
c_{LB}(\tau, \kappa)
    &= \overline{\wtilde{\Psi}}_{LB}(\tau, \kappa) - e^\kappa \overline{\Psi}_{LB}(\tau, \kappa) \\
    &= \Phi_{LB}\big(-z(\tau, \kappa); \beta(\tau) - \varsigma(\tau), \alpha(\tau) + \varsigma(\tau)\big) - e^\kappa \Phi_{LB}\big(-z(\tau, \kappa); \beta(\tau), \alpha(\tau) \big).
\end{aligned}
\end{equation}
To ensure valid densities and to exclude static arbitrage, the following conditions on the term structures should hold \citep[Proposition 4.1]{carr2021additive}: (i) $\varsigma < \beta$; (ii) $\alpha, \beta > 0$; (iii) $\varsigma$ is nondecreasing and $\varsigma(0) = 0$; and (iv) $\frac{\alpha}{\varsigma}$ and $\frac{\beta}{\varsigma}$ are nonincreasing.

For $X \sim LB(\mu, \varsigma, \alpha, \beta)$, the first four central moments, expectation, variance, skewness, and (excess) kurtosis, are given by \citep{balakrishnan1991handbook}:
\begin{equation}\label{eq:lb_moment}
\begin{aligned} 
\E[X] &= \varsigma(\gamma(\alpha) - \gamma(\beta)) + \mu, 
&&&
\V[X] &= \varsigma^2(\gamma'(\alpha) + \gamma'(\beta)), 
\\
\S[X] &= \frac{\gamma''(\alpha) - \gamma''(\beta)}{(\gamma(\alpha) +\gamma(\beta))^{3/2}}, 
&&&
\K[X] &= \frac{\gamma'''(\alpha) + \gamma'''(\beta)}{(\gamma(\alpha) + \gamma(\beta))^2},
\end{aligned}
\end{equation}
where $\gamma^{(n)}$ are the polygamma functions of order $n$.

\section{Implied Volatility Reframed} \label{sec:volatility}
Implied volatility is understood at a practical level as the quantity that matches the option price produced by the Black-Scholes formula~\eqref{eq:pc_n} to the market option price. It changed from a constant to a function of tenors and moneynesses after the 1987 Black Monday crash. That is, to reproduce the same option price from the market, one has to amend pointwise the parameter $\sigma \in \R_+$ with respect to both tenor and moneyness, into $\sigma: \R_+ \times \R \to \R_+$, $(\tau, \kappa) \mapsto \sigma(\tau, \kappa)$. Equations~\eqref{eq:omega_sigma_const}-\eqref{eq:pc_n} have to be modified accordingly. The total implied volatility becomes
\begin{align} \label{eq:omega_sigma_iv}
\omega(\tau, \kappa) := \sigma(\tau, \kappa) \sqrt{\tau},
\end{align}
and the new pivots are
\begin{align} \label{eq:z_n_iv}
z_\pm^\omega(\tau, x; \kappa) := \frac{x \pm \frac{1}{2}\omega^2(\tau, \kappa)}{\omega(\tau, \kappa)},
&&
z_\pm^\omega(\tau, x) := z_\pm^\omega(\tau, x; x).
\end{align}
The pricing formula is reformulated with $\omega$ as
\begin{align} \label{eq:pc_n_iv}
p_{BS}^\omega(\tau, \kappa) = e^\kappa \Phi\big(z_+^\omega(\tau, \kappa)\big) - \Phi\big(z_-^\omega(\tau, \kappa)\big),
&&
c_{BS}^\omega(\tau, \kappa) = \Phi\big(-z_-^\omega(\tau, \kappa)\big) - e^\kappa \Phi\big(-z_+^\omega(\tau, \kappa)\big).
\end{align}
The modified Black-Scholes PDF and CDF are given by
\begin{equation} \label{eq:pdf_n_iv}
\begin{aligned}
\psi_{BS}^\omega(\tau, x; \kappa) 
    &= \frac{1}{\omega(\tau, \kappa)} \phi\big(z_+^\omega(\tau, x; \kappa)\big),
&&&
\psi_{BS}^\omega(\tau, x)
    &= \psi_{BS}^\omega(\tau, x; x),
\\
\Psi_{BS}^\omega(\tau, x; \kappa)
    &= \Phi \big( z_+^\omega(\tau, x; \kappa) \big),
&&&
\Psi_{BS}^\omega(\tau, x)
    &= \Psi_{BS}^\omega(\tau, x; x).
\end{aligned}
\end{equation}
The tilted PDF $\wtilde{\psi}_{BS}^\omega$ and CDF $\wtilde{\Psi}_{BS}^\omega$ are modified in the same way. Noticing that $\psi_{BS}^\omega(\tau, x)$ (resp. $\Psi_{BS}^\omega(\tau, x)$) without a fixed $\kappa$ is not a valid PDF (resp. CDF), we call it the Black-Scholes quasi-PDF (resp. quasi-CDF).

We also insert the Black-Scholes vega $v_{BS}^\omega$ here: 
\begin{equation} \label{eq:vega}
v_{BS}^\omega(\tau, \kappa) := \pd_\sigma p_{BS}^\omega(\tau, \kappa) = \pd_\sigma c_{BS}^\omega(\tau, \kappa) = \phi\big(z_-^\omega(\tau, \kappa)\big) \sqrt{\tau}.
\end{equation}
It is useful in calibration to convert price differences into volatility differences, motivated by the first order Taylor expansion $\sigma_1 - \sigma_0 = (\pd_\sigma p)^{-1} (p_1 - p_0) + O(p_1 - p_0)^2$.

\subsubsection*{Volatility reductionism}
Modern option study relies heavily on the notion of volatility to both form intuition and describe underlying dynamics, such as local volatility, stochastic volatility, local stochastic volatility, and rough volatility. However, from the lens of marginal distribution, the term volatility is unexpressive on the tail behavior that is more relevant to the documented smile shapes of implied volatility. The contrast between the risk-neutral logistic-beta distribution and the risk-neutral normal distribution provides a clear example.
\begin{figure}[ht]
     \centering
        \includegraphics[width=\textwidth]{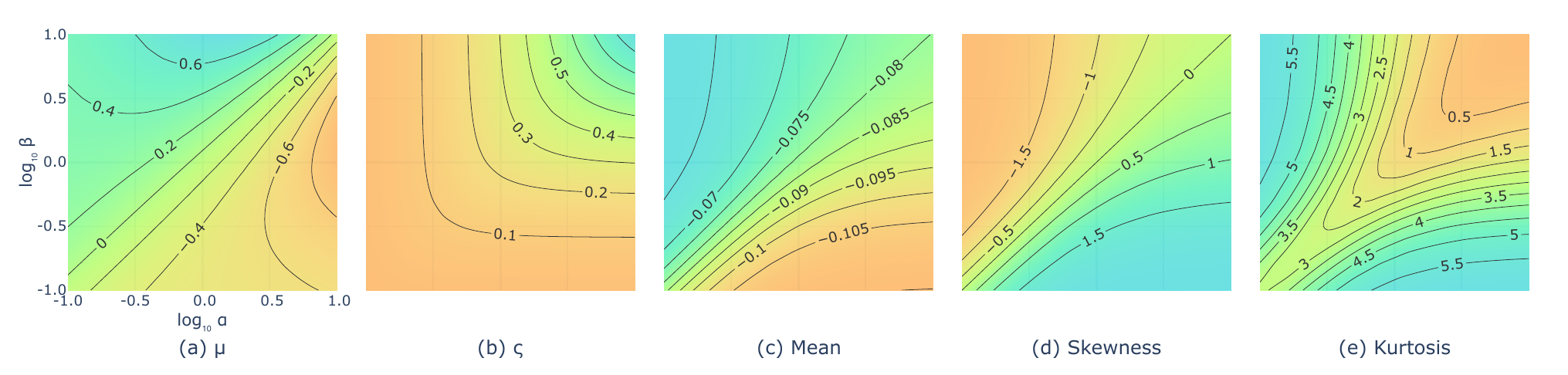}
    \caption{Contour plot of parameters and moments generated by $X \sim LB(\mu, \varsigma, \alpha, \beta)$ with fixed variance $\V[X] = 0.16$. $(\alpha, \beta) \in [10^{-1}, 10^{1}]^2$. $\varsigma$ in (b) is determined by Equation~\eqref{eq:lb_varsigma}, $\mu$ in (b) is determined by Equation~\eqref{eq:lb_mu}, and $\E[X]$, $\S[X]$, and $\K[X]$ in (c)-(e) are determined by Equation~\eqref{eq:lb_moment}.}
    \label{fig:moment_countor}
\end{figure}

\begin{figure}[ht]
     \centering
        \includegraphics[width=\textwidth]{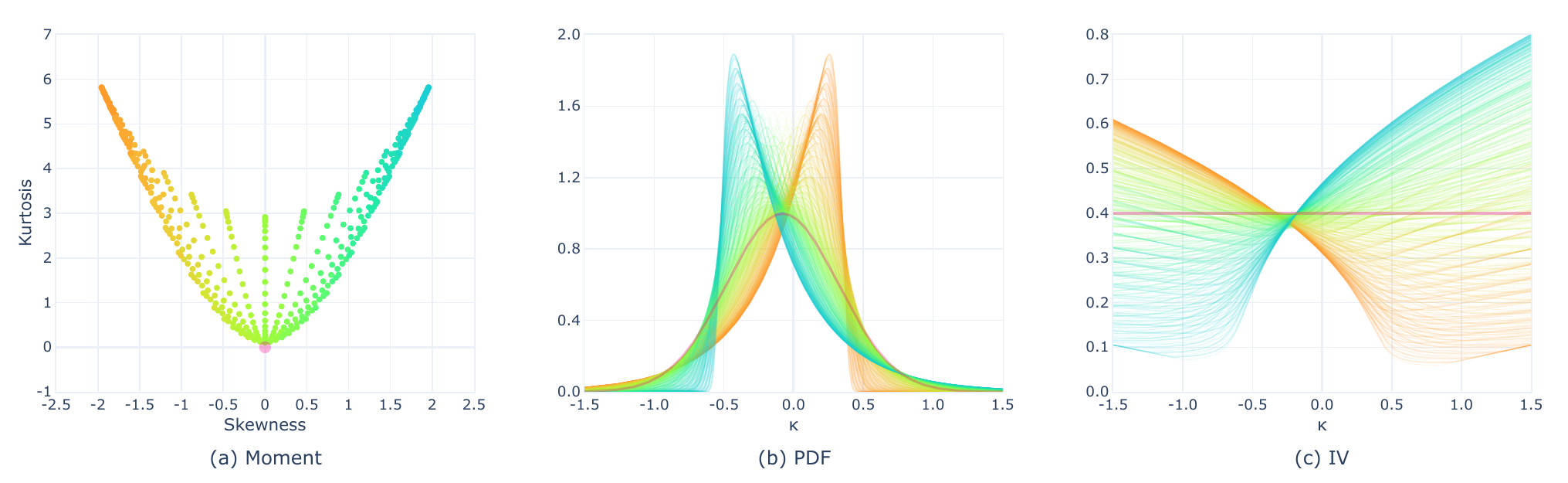}
    \caption{Risk-neutral logistic-beta distributions versus risk-neutral normal distribution. All distributions have a variance of $0.16$. The light violet color corresponds to the normal distribution. $400$ logistic-beta distributions with $(\alpha, \beta) \in \{10^{-1}, 10^{-0.9}, \dots, 10^{1}\}^2$ are plotted, specified in the same manner as Figure~\ref{fig:moment_countor}, colored by their skewnesses. Negative (positive) skewness is in the red (blue) end of the color spectrum.}
    \label{fig:moment_fixvol}
\end{figure}

Suppose for now $\tau = 1$ and we omit the $\tau$ notation. Let us fix the volatility of the return to $0.4$, equivalently $\V[X] = 0.16$. For the Black-Scholes model $X \sim N(-\frac{1}{2} \sigma^2, \sigma^2)$ it only means $\sigma=0.4$, while the additive logistic model $X \ \sim LB(\mu, \varsigma, \alpha, \beta)$ has a wide range of choices. For given $\alpha$ and $\beta$, the dispersion parameter $\varsigma$ is first determined with Equation~\eqref{eq:lb_moment} as
\begin{align} \label{eq:lb_varsigma}
    \varsigma = \left(\frac{\V[X]}{\gamma'(\alpha) + \gamma'(\beta)}\right)^{1/2},
\end{align}
and the location parameter $\mu$ is then calculated by martingality according to Equation~\eqref{eq:lb_mu}. 

Thus, even with a fixed variance, by altering the choice of $\alpha$ and $\beta$, the risk-neutral logistic-beta distribution can generate a wide range of skewness and kurtosis, and consequently various shapes of implied volatility curves. Figure~\ref{fig:moment_countor} displays contour plots of $\mu$, $\varsigma$, $\E[X]$, $\S[X]$, and $\K[X]$ from risk-neutral logistic-beta distributions with $(\alpha, \beta) \in [10^{-1}, 10^{1}]^2$. Figure~\ref{fig:moment_fixvol} collects $400$ risk-neutral logistic-beta distributions with $(\alpha, \beta) \in \{10^{-1}, 10^{-0.9}, \dots, 10^{1}\}^2$ in contrast to the single risk-neutral normal distribution. Figure~\ref{fig:moment_fixvol} (a) shows the resulting combination of skewness and kurtosis pairs, (b) shows their PDFs, and (c) illustrates the rich shapes of the resulting implied volatility curves.

\subsubsection*{Implied volatility conflation}
It is now clear how, without skewness and kurtosis, volatility is impassive for distributional information. From this point of view, we understand that implied volatility is an overloaded term that compresses dispersion, asymmetry, and heavy tails altogether. We henceforth exemplify this by assuming that the true underlying return follows the additive logistic model. The correct risk-neutral density is $X \sim LB(\mu, 0.15, 0.57, 1.15)$, where $\mu$ is calculated by Equation~\eqref{eq:lb_mu} and the other parameters are taken from parameters specified later.

\begin{figure}[ht]
    \centering
        \includegraphics[width=\textwidth]{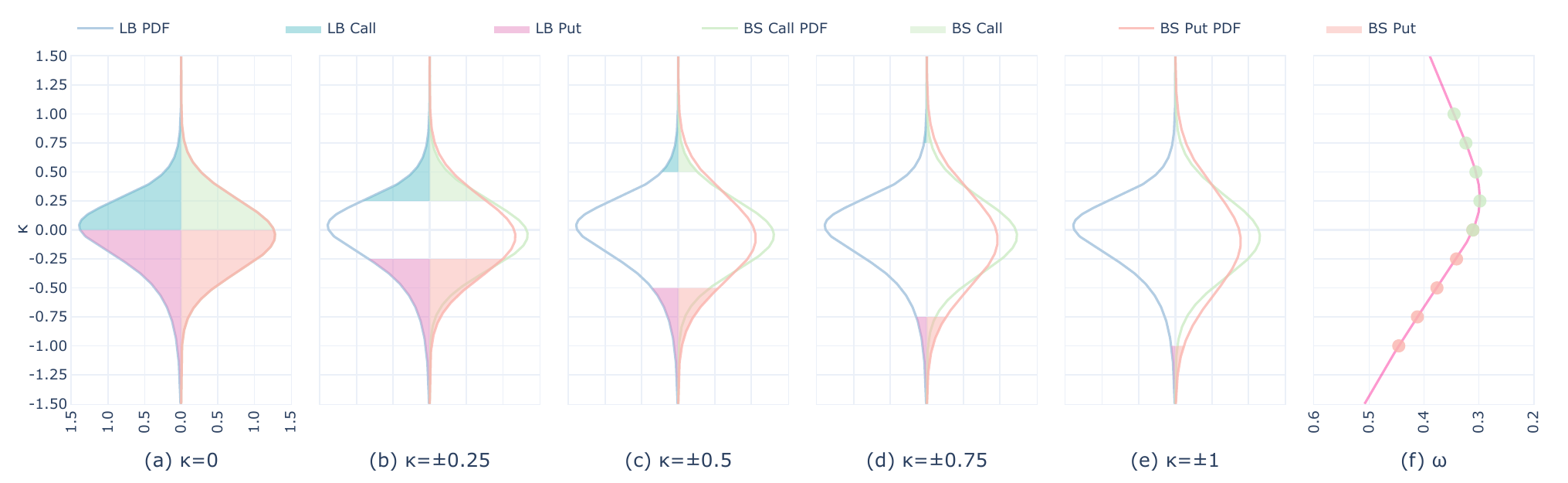}
    \caption{True logistic-beta density versus Black-Scholes densities. In (a)-(e), the consistent blue curve is the PDF $\psi_{LB}$ of $LB(\mu, 0.15, 0.57, 1.15)$, and the green (red) curves are Black-Scholes $\{\psi_{BS}^\omega(\cdot; \kappa)\}_\kappa$ with various implied volatility to match call (put) prices. The blue (purple) shaded areas correspond to the integral area of $\psi_{LB}$ on the supports of calls (puts). The green (red) shaded areas correspond to the integral area of $\{\psi_{BS}^\omega(\cdot; \kappa)\}_\kappa$ on the supports of calls (puts). (f) shows the corresponding implied volatility.}
    \label{fig:pdfiv}
\end{figure}

Figure~\ref{fig:pdfiv} animates how the Black-Scholes PDF is stretched to match option prices produced by the additive logistic model. (a)-(e) contrasts the correct logistic-beta density $\psi_{LB}$ against multiple Black-Scholes densities $\psi_{BS}^\omega(\cdot; \kappa)$ that varies with $\kappa$. Shaded areas correspond to integral areas on the supports of put prices and call prices. The consistent blue curve is the correct $\psi_{LB}$, while the green and red curves representing $\{\psi_{BS}^\omega(\cdot; \kappa)\}$ with $\kappa \in \{0, \pm0.25, \pm0.5, \pm0.75, \pm1\}$ that are stretched to match option prices. Because $\psi_{LB}$ is asymmetric and heavy-tailed, as $\kappa$ moves away from zero to both tails, $\psi_{BS}^\omega(\cdot; \kappa)$ spreads flatter at different paces. The consequent asymmetric and smile-like implied curves in (f) are a footprint indicating how $\psi_{BS}^\omega(\cdot; \kappa)$ has to alter its dispersion to compensate for the lack of skewness and kurtosis.

\subsubsection*{Implied scale as corrector}
Based on the above observation, a more holistic interpretation of the implied volatility surface can be inferred: rather than volatility, implied volatility is a correction function to pointwisely morph the Black-Scholes quasi-density into the correct density. We formalize this idea by the following proposition.

\begin{proposition} \label{prop:psi_omega}
Let $\psi$ and $\Psi$ be the market risk-neutral PDF and CDF, $p$ be the relative put option price defined in Equation~\eqref{eq:pc}, $\omega$ be the total implied volatility surface defined in Equation~\eqref{eq:omega_sigma_iv}. The following equalities hold:
\begin{align}
\pd_\tau p(\tau, \kappa)
    &= \wtilde{\psi}_{BS}^\omega(\tau, \kappa) \omega(\tau, \kappa) \pd_\tau \omega(\tau,\kappa),
    \label{eq:dp_omega}
\\
\Psi(\tau, \kappa)
    &= \Psi_{BS}^\omega(\tau, \kappa) + \zeta^\omega(\tau, \kappa),
    \label{eq:Psi_omega}
\\
\psi(\tau, \kappa)
    &= \psi_{BS}^\omega(\tau,\kappa) \xi^\omega(\tau, \kappa),
    \label{eq:psi_omega}
\end{align}
where $\zeta^\omega$ is the correction addend and $\xi^\omega$ is the correction multiplier given by
\begin{equation} \label{eq:zeta_xi}
\begin{aligned}
\zeta^\omega(\tau, \kappa)
    &= \psi_{BS}^\omega(\tau, \kappa) \omega(\tau, \kappa) \pd_\kappa \omega(\tau, \kappa),
\\
\xi^\omega(\tau, \kappa)
    &= \left(1 - \frac{\kappa}{\omega(\tau, \kappa)} \pd_\kappa \omega(\tau, \kappa) \right)^2 - \frac{1}{4} \big(\omega(\tau, \kappa) \pd_\kappa \omega(\tau, \kappa)\big)^2 + \omega(\tau, \kappa) \pd_{\kappa \kappa} \omega(\tau, \kappa).
\end{aligned}
\end{equation}
$\psi_{BS}^\omega$, $\wtilde{\psi}_{BS}^\omega$, and $\Psi_{BS}^\omega(\tau, \kappa)$ are the Black-Scholes quasi-PDF, tilted quasi-PDF, and quasi-CDF defined around Equation~\eqref{eq:pdf_n_iv}.
\end{proposition}
Equation~\eqref{eq:dp_omega} relates to the term structure of $\omega$. In Equation~\eqref{eq:Psi_omega}, the addend $\zeta^\omega$ pointwisely matches the Black-Scholes quasi-CDF with the true CDF. Equation~\eqref{eq:psi_omega} corrects the Black-Scholes quasi-PDF with the true PDF through the multiplier $\xi^\omega$. Figure~\ref{fig:div} visualizes all of these related quantities.

\begin{figure}[ht]
     \centering
        \includegraphics[width=\textwidth]{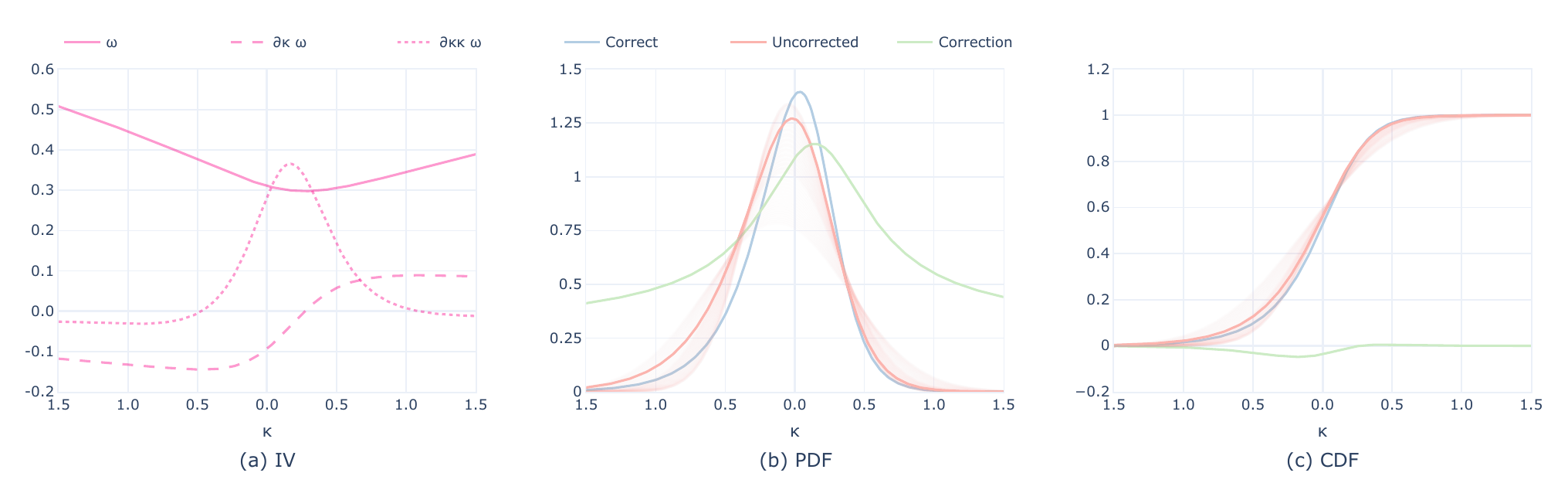}
    \caption{Implied volatility as pointwise correction. The correct distribution is $X_1 \sim LB(\mu, 0.15, 0.57, 1.15)$. (a) shows the corresponding implied total volatility $\omega$ and its spatial partial derivatives $\pd_\kappa \omega$ and $\pd_{\kappa \kappa} \omega$. In (b) (resp. (c)), blue curve is the correct PDF $\psi_{LB}$ (resp. CDF $\Psi_{LB}$), red curve is the uncorrected Black-Scholes quasi-PDF $\psi_{BS}^\omega$ (resp. quasi-CDF $\Psi_{BS}^\omega$), and green curve is the corrector $\zeta^\omega(1,x)$ (resp. $\xi^\omega)$. We additionally draw the related Black-Scholes PDFs $\{\psi_{BS}(\cdot; \kappa)\}_\kappa$ (resp. CDFs $\{\Psi_{BS}(\cdot; \kappa)\}_\kappa$) in very light red curves.}
    \label{fig:div}
\end{figure}

\begin{remark}
In the same spirit of seeing $\omega$ as the corrector, we can imagine other possible pricing formulas with correctors and other analogs of Proposition~\ref{prop:psi_omega}. For example, a degenerate additive logistic pricing formula~\eqref{eq:pc_lb} with $\alpha=\beta=1$ that is corrected by an implied scale surface $\varsigma(\tau, \kappa)$: $p_{LB}^\varsigma(\tau, \kappa) = e^\kappa \Phi_L \big(z^\varsigma(\tau, \kappa) \big) - \Phi_{LB}\big( z^\varsigma(\tau, \kappa); 1 + \varsigma(\tau, \kappa), 1 - \varsigma(\tau, \kappa) \big)$. 
\end{remark}

\subsubsection*{Sufficient conditions on implied volatility to exclude static arbitrage}

Combining Proposition~\ref{prop:psi_omega}, Equation~\eqref{eq:con_arb_p}, and \eqref{eq:pd_p} readily conveys the risk-neutrality of implied density to sufficient conditions imposed on implied volatility to exclude static arbitrages, which we casually collect in the following corollary. For more detailed treatments, we refer to \citep[Theorem 2.9]{roper2010arbitrage}.
\begin{corollary} \label{coro:no_arb}
An implied volatility surface $\omega$ admits
\begin{enumerate}[i.]
    \item no calendar spread arbitrage if, for every $\kappa \in \R$, $\omega(\cdot, \kappa)$ is nondecreasing;
    \item no vertical spread arbitrage if, for every $\tau \in \R_+$, $\Psi_{BS}^\omega(\tau, \cdot) + \zeta^\omega(\tau, \cdot)$ is a valid CDF;
    \item no butterfly spread arbitrage if, for every $\tau \in \R_+$, $\psi_{BS}^\omega(\tau, \cdot) \xi^\omega(\tau, \cdot)$ is a valid PDF.
\end{enumerate}
\end{corollary}
Because we only focus on data from the bound region $(\tau, \kappa) \in [\underline{\tau}, \overline{\tau}] \times [\underline{\kappa}, \overline{\kappa}]$ in the next section, we leave out asymptotics of expiring tenor and extreme moneyness. Hence, the conditions we need from Corollary~\ref{coro:no_arb} boil down to the following:
\begin{align} \label{eq:con_arb_omega}
\pd_\tau \omega(\tau,\kappa) \ge 0,
&&
\zeta^\omega(\tau, \kappa) \ge -\Psi_{BS}^\omega(\tau, \kappa),
&&
\xi^\omega(\tau, \kappa) \ge 0.
\end{align}
The nonlinearity of $\zeta^\omega$ and $\xi^\omega$ still hinders direct parameterization of $\omega$, which stimulates the use of a neural network to approximate $\omega$.

\section{Neural Representation Shallowed} \label{sec:neural}

The expressibility of neural networks has been extensively studied in the past decades. \citep{hornik1989multilayer, hornik1990universal, hornik1991approximation} provided very general conditions for standard multilayer feedforward networks with sufficiently smooth activation to well approximate functions and their derivatives with respect to Lebesgue space. Although theoretically verified, expressibility relies on a particular design of neural network structure, including depth, width, and choice of activation functions. For instance, \citep{eldan2016power} prioritizes depth of network over width for better expressibility, while \citep{cabanilla2024neural} shows that one can reduce the depth of the network with an appropriate choice of activation functions. Nonetheless, building a neural representation of implied volatility and implied density falls outside the standard setting considered above; its viability should be tested with extensive experiments.

In contrast to previous successful implementations in \citep{ackerer2020deep, zheng2021incorporating, wiedemann2024operator, yang2025hyperiv} with deep or sophisticated neural networks, our goal is to find the simplest neural representations for option implied information upon the representative additive logistic model. It is motivated by the following rationales. Firstly, both implied density and implied volatility are, after all, smooth bivariate functions, and they seem to have been consistently behaving under the usual market conditions for a long time \citep{dotsis2025option}. So, we expect that a simple neural network is sufficient for expressibility. Second, the additive logistic model is highly customizable for both implied density and implied volatility, and is representative of usual market conditions \citep{lin2024neural}. We suggest that it is preliminary to fit the neural representation before real data; a more holistic experiment can be conducted, such as goodness of fit to implied density that is not accessible in the real situation; the well-trained neural representation can be the initial guess for learning with real market data. With such a goal, we design the neural representation framework and the experiments.

\subsubsection*{Framework}
Let $\what{f}$ denote the neural network approximation of function $f$. The neural representation framework feeds on an option chain $\kD$ as the input data to learn the implied density $\what{\psi}(\cdot, \cdot)$ that is backed by the implied total volatility $\womega(\cdot, \cdot)$. The framework of neural representation of option implied information consists of a simple feedforward neural network plus several auto-differentiable modules, outlined in Figure~\ref{fig:framework}. We now describe each module.

\begin{figure}[ht]
     \centering
        \includegraphics[width=0.8\textwidth]{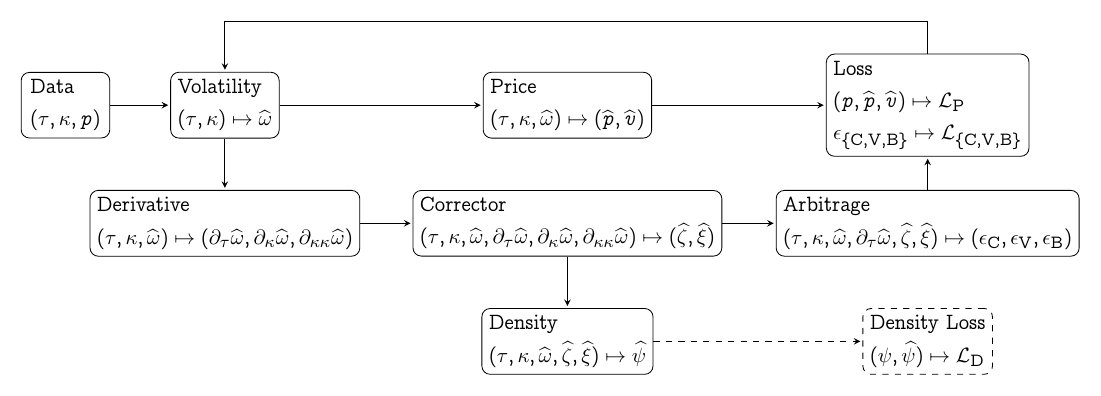}
    \caption{Neural representation framework.}
    \label{fig:framework}
\end{figure}

\textbf{Data.}
The dataset includes tenors, moneynesses, and the corresponding option prices in the option chain. We denote the dataset by $\kD = \{(\tau, \kappa, p(\tau, \kappa)\}_{(\tau, \kappa) \in \kT \times \kK}$, where the tenor set $\kT = \{\tau_1, \tau_2, \dots, \tau_N\}$ contains $\abs{\kT} = N$ number of tenors with uniform grid difference $\Delta \tau$ and the moneyness set $\kK =\{\kappa_1, \kappa_2, \dots, \kappa_M\}$ has $\abs{\kK} = M$ number of moneynesses with uniform grid difference $\Delta \kappa$. There are $\abs{\kD} = \abs{\kT} \abs{\kK} = NM$ data points. In the real scenario, the data grid is uneven - number of listed moneynesses varies among different tenors - so here, setting the data points into the Cartesian product of tenor set and moneyness set is a simplification.

\textbf{Volatility module: $(\tau, \kappa) \mapsto \womega(\tau, \kappa)$.}
We approximate the total implied volatility with a feedforward neural network $\womega: \R^2 \to \R_+$ with $L$ hidden layers and an output layer indexed by $L+1$:
\begin{align} \label{eq:cV_comp}
\womega(\tau, \kappa) = \cV(\tau, \kappa) = \cA_{L+1} \circ \cF_{L+1} \circ \cdots \cA_1 \circ \cF_1(\tau, \kappa)
\end{align}
where, for $1 \le l \le L+1$, $\cF_l: \R^{n_{l-1}} \to \R^{n_l}$ that $\cF_l(x_{l-1}) = W_l x_{l-1} + b_l$ is the affine map with $W_l \in \R^{n_l} \times \R^{n_{l-1}}$ being the weights and $b_l \in \R^{n_l}$ being the bias, $\cA_l: \R^{n_l} \to \R^{n_l}$ is the element-wise activation function, and $n_l$ is the width of the layer. $l=0$ corresponds to the input $(\tau, \kappa)$ of dimension $n_0 = 2$. The output layer $l = L+1$ gives a scalar so $n_{L+1} = 1$. As the implied volatility is nonnegative, the last activation function should output nonnegative values. 

\textbf{Price module: $(\tau, \kappa, \womega) \mapsto \big(\what{p}(\tau, \kappa), \what{v}(\tau, \kappa)\big)$.}
With total implied volatility $\womega$ calculated, we calculate the approximated option price $\what{p}$ and in addition the vega $\what{v}$ by plugging $\womega$ into Equation~\eqref{eq:pc_n} and Equation~\eqref{eq:vega}:
\begin{align} \label{eq:p_hat}
\what{p}(\tau, \kappa) = p_{BS}^{\womega}(\tau, \kappa),
&&
\what{v}(\tau, \kappa) = v_{BS}^{\womega}(\tau, \kappa).
\end{align}

\textbf{Derivative module: $(\tau, \kappa, \womega) \mapsto \big(\pd_\tau \womega(\tau, \kappa), \pd_\kappa \womega(\tau, \kappa), \pd_{\kappa \kappa} \womega(\tau, \kappa)\big)$.}
It is easier to express partial derivatives of $\womega$ starting from the recursive form of Equation~\eqref{eq:cV_comp}: 
\begin{equation} \label{eq:cV_rec}
\begin{cases}
x_0 = (\tau, \kappa),
\\
x_l = \cA_l(y_l),
\quad
y_l = \cF_l(x_{l-1}),
\quad
1 \le l \le L + 1,
\\
\womega = x_{L+1}.
\end{cases}
\end{equation}
Differentiating Equation~\eqref{eq:cV_rec} we obtain the first order partial derivatives $\pd_\tau \womega$ and $\pd_\kappa \womega$,
\begin{equation} \label{eq:d_omega_hat}
\begin{cases}
\pd_\cdot \womega = \pd_\kappa x_{L + 1}
\\
\pd_\cdot x_l = \cA_{l}'(y_{l}) \odot (W_l \pd_\cdot x_{l-1}),
\quad
1 \le l \le L + 1,
\\
\pd_\cdot x_0 =(\pd_\cdot \tau, \pd_\cdot \kappa), 
\end{cases}
\end{equation}
and the second order partial derivative $\pd_{\kappa \kappa} \womega$,
\begin{equation} \label{eq:dd_omega_hat}
\begin{cases}
\pd_{\kappa \kappa} \womega = \pd_{\kappa \kappa} x_{L + 1},
\\
\pd_{\kappa \kappa} x_l = \cA_{l}''(y_{l}) \odot (W_l \pd_\cdot x_{l-1})^{\odot 2} + \cA_{l}'(y_{l}) \odot (W_l \pd_{\kappa \kappa} x_{l-1}),
\quad
1 \le l \le L + 1,
\\
\pd_{\kappa \kappa} x_0 = 0,
\end{cases}
\end{equation}
where $\odot$ is Hadamard product, $\cA'$ and $\cA''$ are the element-wise first and second derivatives of $\cA$.

\textbf{Corrector module: $(\tau, \kappa, \womega, \pd_\tau \womega, \pd_\kappa \womega, \pd_{\kappa \kappa} \womega) \mapsto  \big(\what{\zeta}(\tau, \kappa), \what{\xi}(\tau, \kappa)\big)$.} The correctors are calculated by plugging $\womega$ into Equation~\eqref{eq:zeta_xi}:
\begin{align}
\what{\zeta}(\tau, \kappa) = \zeta^{\womega}(\tau, \kappa),
&&
\what{\xi}(\tau, \kappa) = \xi^{\womega}(\tau, \kappa).
\end{align}

\textbf{Density module: $(\tau, \kappa, \womega, \what{\zeta}, \what{\xi}) \mapsto \what{\psi}(\tau, \kappa)$.} The estimated implied density $\what{\psi}$ is calculated by substituting $\womega$ into Equation~\eqref{eq:psi_omega} and \eqref{eq:zeta_xi}:
$$\what{\psi}(\tau, \kappa) = \psi_{BS}^{\womega}(\tau, \kappa) \xi^{\womega}(\tau, \kappa).$$

\textbf{Arbitrage module: $(\tau, \kappa, \womega, \pd_\tau \womega, \what{\zeta}, \what{\xi}) \mapsto \big(\epsilon_C(\tau, \kappa), \epsilon_V(\tau, \kappa), \epsilon_B(\tau, \kappa)\big)$.} Errors corresponding to calendar spread arbitrage, vertical spread arbitrage, and butterfly spread arbitrage are
\begin{align} \label{eq:e_arb}
\epsilon_C(\tau, \kappa)
    = \( - \pd_\tau \womega(\tau, \kappa)\)^+,
&&
\epsilon_V(\tau, \kappa)
    = \( - \Psi_{BS}^{\womega}(\tau, \kappa) - \zeta^{\womega}(\tau, \kappa) \)^+,
&&
\epsilon_B(\tau, \kappa)
    = \( - \xi^{\womega}(\tau, \kappa) \)^+,
\end{align}
where $\zeta$ and $\xi$ are calculated based on Equation~\eqref{eq:zeta_xi} with $\womega$ plugged in. $\epsilon_{\{C, V, B\}}$ quantifies the degree to which inequalities~\eqref{eq:con_arb_omega} are violated, and any feasible model should ensure they are equal to zero.

\textbf{Loss module.} 
Evaluated over the dataset $\kD$, the total loss function $\cL(\kD)$ is the sum of losses caused by pricing loss, $\cL_P(\kD)$, and arbitrage violations, $\cL_C(\kD)$, $\cL_V(\kD)$, $\cL_B(\kD)$, calculated as
\begin{equation} \label{eq:L}
\cL(\kD) = \underbrace{\cL_P(\kD)}_{\text{weighted price}} 
    + \underbrace{\cL_C(\kD)}_{\text{calendar spread}}
    + \underbrace{\cL_V(\kD)}_{\text{vertical spread}}
    + \underbrace{\cL_B(\kD)}_{\text{butterfly spread}}
\end{equation}
where each loss is chosen to be the root mean square error as
\begin{equation} \label{eq:L_def}
\begin{aligned}
\cL_P(\kD)
    &= \(\frac{1}{\abs{\kD}} \sum_{(\tau,\kappa) \in \kT \times \kK} \frac{1}{\what{v}(\tau, \kappa)}\big(\what{p}(\tau, \kappa) - p(\tau,\kappa)\big)^2\)^{1/2},
\\
\cL_{\{C, V, B\}}(\kD)
    &= \(\frac{1}{\abs{\kD}} \sum_{(\tau,\kappa) \in \kT \times \kK} \big(\epsilon_{\{C, V, B\}}(\tau, \kappa) \big)^2\)^{1/2},
\end{aligned}    
\end{equation}
A valid trained model should hold that
\begin{equation}
\cL_{\{C, V, B\}}(\kD) \equiv 0.
\end{equation}
Also notice that the vega weighted price loss $\cL_P$ is equivalent to the implied volatility loss.

\textbf{Extra density loss module.}
The market risk-neutral distribution is not observable in practice. We exploit the known ground truth density in the synthetic data experiment to examine the goodness of fit to the implied density of the neural representation. Density deviation is measured with the density loss $\cL_D$ over the dataset $\kD$ that is designed as
\begin{align}
\cL_D(\kD) = \(\frac{1}{\abs{\kT}} \sum_{(\tau, \kappa) \in \kT \times \kK} \big(\what{\psi}(\tau, \kappa) - \psi(\tau, \kappa)\big)^2 \Delta \kappa \)^{1/2}.
\end{align}
In the later section, we will report a case where the $\what{p}$ and $\what{\sigma}$ well approximate $p$ and $\sigma$, but the $\what{\psi}$ significantly deviates from $\psi$. 

\subsubsection*{Experiment Specification}
\textbf{Synthetic market model.} We select the additive logistic model as the synthetic model for the market, with the following term structure that was tested in \citep{azzone2025explicit, lin2024neural}:
\begin{align} \label{eq:lb_term}
\varsigma(\tau) = \varsigma_0 \tau^{h_0}, 
&&
\alpha(\tau) = \alpha_1 + \frac{\alpha_0 - \alpha_1}{1 + \varsigma(\tau)},
&&
\beta(\tau) = \beta_1 + \frac{\beta_0 - \beta_1}{1 + \varsigma(\tau)} + \varsigma(\tau).
\end{align}
We choose the following parameters to produce the typical left-skewed implied density and implied volatility:
\begin{align} \label{eq:lb_term_param}
\varsigma_0 = 0.15,
&&
h_0 = 0.5,
&&
\alpha_0 = 0.5,
&&
\alpha_1 = 1,
&&
\beta_0 = 1,
&&
\beta_1 = 1.
\end{align}
Figure~\ref{fig:lb} illustrates the corresponding implied density surface and the implied volatility surface.

\begin{figure}[ht]
     \centering
     \begin{subfigure}[b]{0.4\textwidth}
         \centering
         \includegraphics[width=\textwidth, trim=1cm 0cm 1cm 2cm, clip=true]{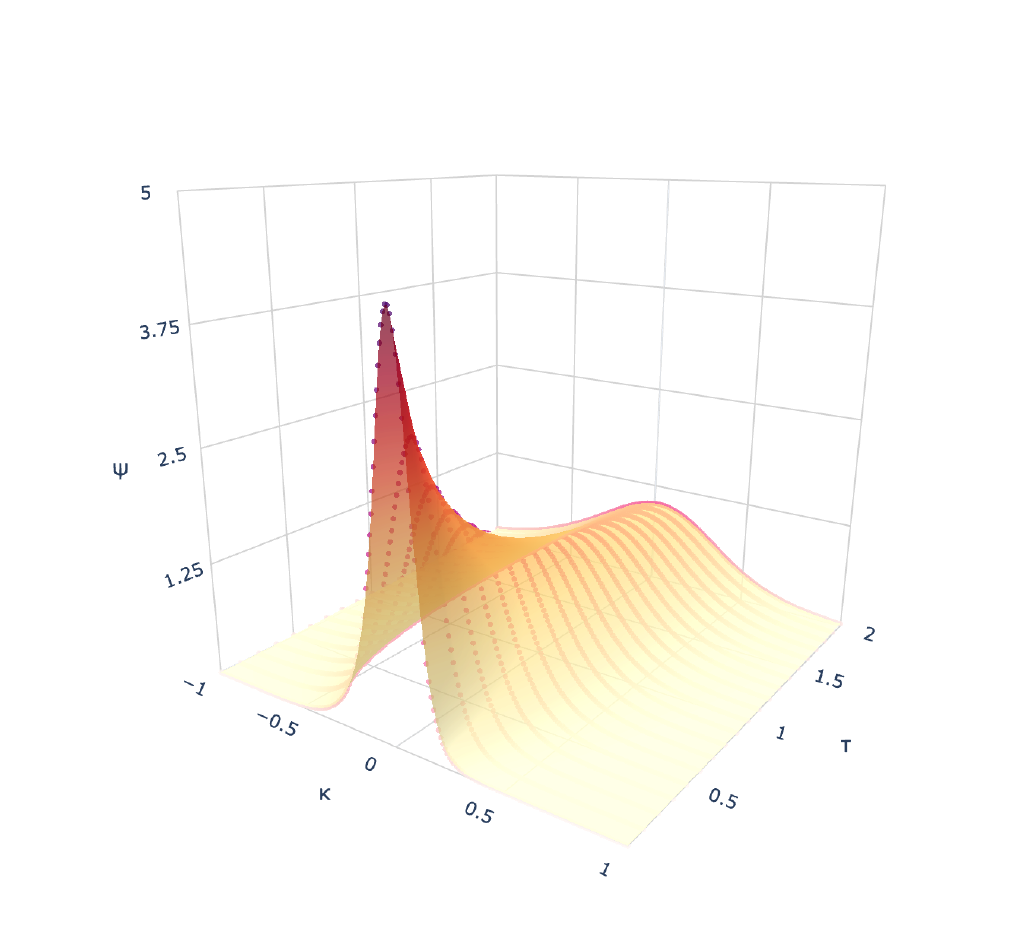}
         \caption{Implied density}
     \end{subfigure}
     \begin{subfigure}[b]{0.4\textwidth}
         \centering
         \includegraphics[width=\textwidth, trim=1cm 0cm 1cm 2cm, clip=true]{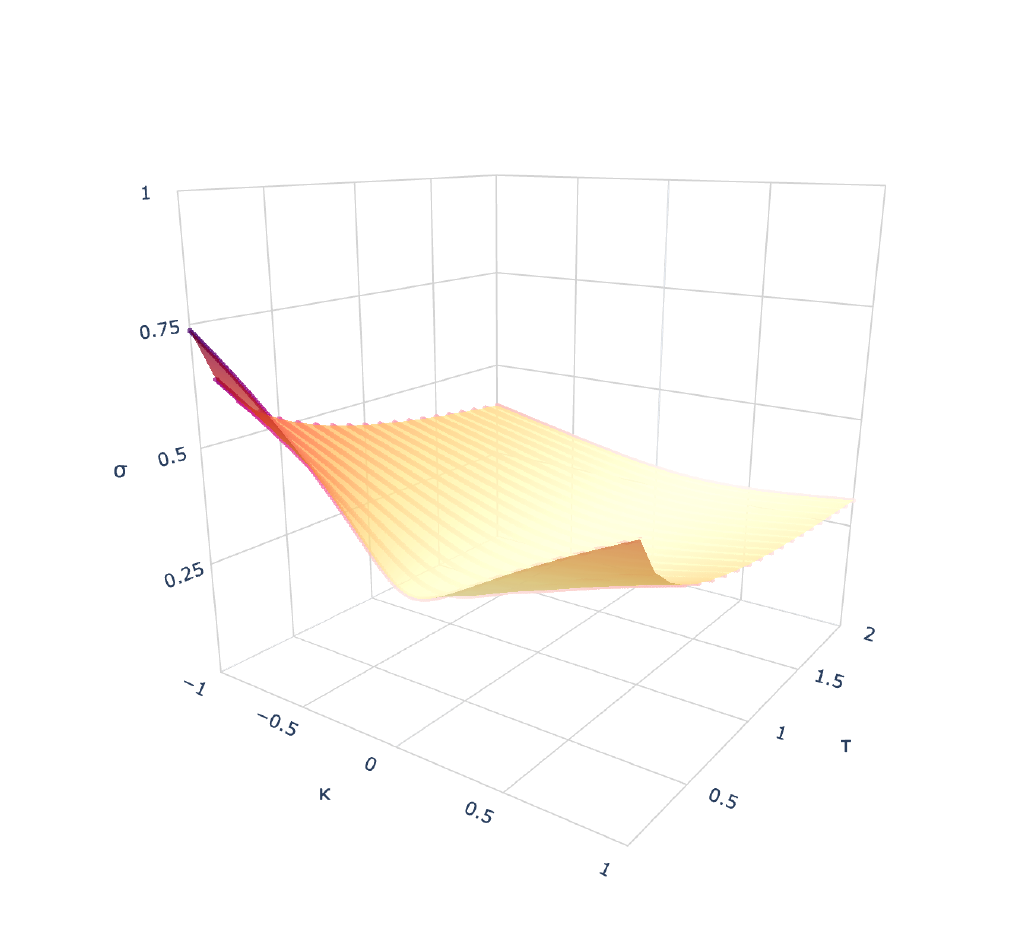}
         \caption{Implied volatility}
     \end{subfigure}
    \caption{Synthetic surfaces generated by additive logistic model with term structure in Equation~\eqref{eq:lb_term} and \eqref{eq:lb_term_param}. Red scatters are data points of the training set $\kD_T$.}
    \label{fig:lb}
\end{figure}

\textbf{Data.} 
The experiment contains two datasets: the training set $\kD_T$ to imitate the observable option chain in the real scenario, and the validation set $\kD_V$, which contains denser points for investigating the model quality. The training set has $20$ tenors in $\kT_T = \{0.1, 0.2, \dots, 2\}$ with $\Delta \tau_T = 0.1$ and $201$ moneynesses in $\kK_T = \{-1, -0.99, \dots, 1\}$ with $\Delta \kappa_T = 0.01$, so $\abs{\kD_T} = 4,020$. The validation set contains $191$ tenors in $\kT_V = \{0.1, 0.11, \dots, 2\}$ with $\Delta \tau_V = 0.01$ and $2010$ moneynesses in $\kK_T = \{-1, -0.999, \dots, 1\}$ with $\Delta \kappa_T = 0.001$, so $\abs{\kD_V} = 383,910$. Thus, the validation set is a finer partition of the training set with a ratio of $\abs{\kD_T}:\abs{\kD_V} = 1:95.5$. 

With the mesh grid chosen, the option prices are calculated with Equation~\eqref{eq:pc_lb}. So far, we explain the methodology mainly in terms of the put option $p$, as it is mathematically equivalent through the put-call parity $c - p = 1 - e^\kappa$. Yet for the sake of numerical precision in actual calculation, we shall instead use the out-of-the-money option prices defined by 
\begin{align} \label{eq:o}
o(\tau, \kappa) := \one_{\{\kappa \le 0\}} p(\tau, \kappa) + \one_{\{\kappa > 0\}} c(\tau, \kappa).
\end{align}
This is a simple adaptation obtained by replacing $p$ (resp. $\what{p}$) with $o$ (resp. $\what{o}$) in Equation~\eqref{eq:p_hat} and \eqref{eq:L_def}.

\textbf{Neural network structure.}
The core module in the system is the volatility module $\cV$ defined in Equation~\eqref{eq:cV_comp}. We start with the choice of its activation functions. Concerning the nonnegativity of $\omega$, an intuitive choice of the activation function in the last layer is the softplus activation $\cA_{L+1} = \cA_{\text{Softplus}}$, whose formula and derivatives are
\begin{align}
\cA_{\text{Softplus}}(x) = \ln(1 + e^x),
&&
\cA_{\text{Softplus}}'(x) = \frac{e^x}{1 + e^x},
&&
\cA_{\text{Softplus}}''(x) = \frac{e^x}{(1 + e^x)^2}.
\end{align}
Notice that the first and second derivatives of the softplus activation coincide with the logistic CDF and PDF given in Equation~\eqref{eq:pdf_l}, which potentially regulates the values of $\pd_\kappa \womega$ and $\pd_{\kappa \kappa} \womega$ to the scale of the CDF and PDF, as suggested by Equation~\eqref{eq:Psi_omega} and \eqref{eq:psi_omega}.

The twice differentiability of $\womega$ is backed by the softplus activation, so we have a wide choice of activations for hidden layers. For simplicity, we let all $\{\cA_{l}\}_{1 \le l \le L}$ have the same activation, and choose from the following candidates: ReLU, quadratic ReLU, cubic ReLU, ELU, and Tanh.
\begin{equation}
\begin{aligned}
&\cA_{\text{ReLU}}(x) = (x)^+,
&
&\cA'_{\text{ReLU}}(x) = \one_{\{x > 0\}},
&
&\cA''_{\text{ReLU}}(x) = 0,
\\
&\cA_{\text{ReLU2}}(x) = \frac{1}{2}((x)^+)^2,
&
&\cA_{\text{ReLU2}}'(x) = (x)^+,
&
&\cA_{\text{ReLU2}}''(x) = \one_{\{x > 0\}},
\\
&\cA_{\text{ReLU3}}(x) = \frac{1}{6}((x)^+)^3,
&
&\cA_{\text{ReLU3}}'(x) = \frac{1}{2}((x)^+)^2,
&
&\cA_{\text{ReLU3}}''(x) = (x)^+,
\\
&\cA_{\text{ELU}}(x) = (e^x - 1) \one_{\{x \le 0\}} + x \one_{\{x > 0\}},
&
&\cA_{\text{ELU}}'(x) = e^x \one_{\{x \le 0\}} + \one_{\{x > 0\}},
&
&\cA_{\text{ELU}}''(x) = e^x \one_{\{x \le 0\}},
\\
&\cA_{\text{Tanh}}(x) = \tanh(x),
&
&\cA_{\text{Tanh}}'(x) = \sech^2(x),
&
&\cA_{\text{Tanh}}''(x) = -2 \tanh(x) \sech^2(x).
\end{aligned}
\end{equation}
The motivation for testing these activation functions is as follows. $\cA_{\text{ReLU}}$ has the simplest form and derivatives that simplifies the calculation of $\pd_\cdot \womega$ and $\pd_{\kappa \kappa} \womega$ in Equation~\eqref{eq:d_omega_hat} and \eqref{eq:dd_omega_hat}. Testing $\cA_{\text{ReLU2}}$ and $\cA_{\text{ReLU3}}$ is motivated by \citep{cabanilla2024neural} to facilitate discovery of possible shallow networks. $\cA_{\text{ELU}}$ slightly enhances the smoothness of $\cA_{\text{ReLU}}$. $\cA_{\text{Tanh}}$ is smooth and, in contrast to the other four, bounded.

Noticing the trade-off between the expressivity of the neural network and the increasing complexity of the nested derivative in Equation~\eqref{eq:d_omega_hat} and \eqref{eq:dd_omega_hat} with respect to depth, we hence consider only shallow neural networks with no more than three hidden layers. As for the width of each layer, we select $32$, $64$, and $128$ neurons corresponding to small, medium, and large widths in our context. 

For naming convention, we use $\cN_a^{n, l}$ to refer to neural representation framework as explained in Section~\ref{sec:neural} whose volatility module that consists of corresponding $l$ number of hidden layers of width $n$ and activation function $a$, whereas notation $\cN_a$ without superscript refers to the collection of all the $a$-activation type of neural representations. For example, $\cN_{\text{ReLU}}^{1, 32}$ is the neural representation whose volatility module has one hidden layer of $32$ nodes and $\cA_{\text{ReLU}}$ activation; $\cN_{\text{ReLU}}$ refers to all the ReLU type neural representations. There are in total $45$ models to test:
\begin{equation} \label{eq:nn_set}
\left\{ \cN_a^{n,l}: (a, n, l)\in \{\text{ReLU}, \text{ReLU2}, \text{ReLU3}, \text{ELU}, \text{Tanh}\} \times \{32, 64, 128\} \times \{1, 2, 3\}\right\}.
\end{equation}

\textbf{Training.}
It turns out that the conventional configuration already serves our training purpose. The weights of neural network are initialized with He uniform \citep{he2015delving}, where $\{W_l, b_l\}$ are uniformly distributed between $-n_{l-1}^{-1/2}$ and $n_{l-1}^{-1/2}$ where $n_{l-1}$ is the input dimension of layer $l$. Optimization uses the Adam optimizer \citep{kingma2014adam} with learning rate $10^{-4}$, betas $0.9$ and $0.99$, and epsilon $10^{-16}$, iterating with shuffled mini-batch gradient scheme over $1000$ epochs. We use a batch size of $256$, so there are $\lceil \abs{\kD_T} / 256 \rceil = 16$ shuffled mini-batches, denoted by $\{\kD_T^b\}_{b=1}^{16}$ such that $\bigcup_{b=1}^{16} \kD_T^b = \kD_T$. Since we adopt mini-batch gradient descent, in addition to the losses evaluated over entire datasets, for each epoch of training, we record epoch-wise losses defined as
\begin{align} \label{eq:L_epoch}
\cL^e = \cL_P^e + \cL_C^e + \cL_V^e +\cL_B^e,
&&
\cL_{\{P, C, A, B\}}^e = \frac{1}{16} \sum_{b=1}^{16}\cL_{\{P, C, A, B\}}(\kD_T^b),
\end{align}
which are simple averages of losses on each mini-batch that are also calculated with Equation~\eqref{eq:L_def}.

\subsubsection*{Results and diagnosis}
We first note that all trained models, except for a few ones that fail to converge, are free of arbitrage, with $\cL_{C, V, B}^e(\kD_T) \leadsto 0$ easily achieved within at most $300$ epochs. We also verify that $\cL_{C, V, B}(\kD_V) = 0$. So, the total loss eventually consists of only price loss $\cL \leadsto \cL_P$. Figure~\ref{fig:epoch_loss} reports learning curves of epoch-wise total losses defined in Equation~\eqref{eq:L_epoch} of the $45$ neural representations specified in Equation~\eqref{eq:nn_set}. Losses are reported in basis points (bps) throughout this section, $1 \text{bps} = 10^{-4}$. The first notable observation is the failure of the ReLU3 models. All deep ReLU3 models, $\cN_{\text{ReLU3}}^{\{32, 64, 128\}, \{2, 3\}}$, do not converge. The only exception is the single and wide ReLU3 models $\cN_{\text{ReLU3}}^{\{64, 128\}, 1}$, which appear to learn, but only achieve mediocre performance. Performance of $\cN_{\text{ReLU2}}$ varies. Depending on the network structure, there are not yet converged $\cN_{\text{ReLU2}}^{32, 3}$, bottom performer $\cN_{\text{ReLU2}}^{32, 2}$, and top performer $\cN_{\text{ReLU2}}^{128, 1}$. We will take a closer look at the pattern of $\cN_{\text{ReLU2}}$ shortly. $\cN_{\text{ELU}}$ and $\cN_{\text{Tanh}}$ generally perform well. We shall highlight a special case. $\cN_{\text{ReLU}}$ are most plausible as all models have training loss less than $15$ bps, and some of them achieve the lowest loss around $4$ bps. Unfortunately, further investigation shows that this remarkable price fit is misleading. 

\begin{figure}[ht]
     \centering
        \includegraphics[width=\textwidth]{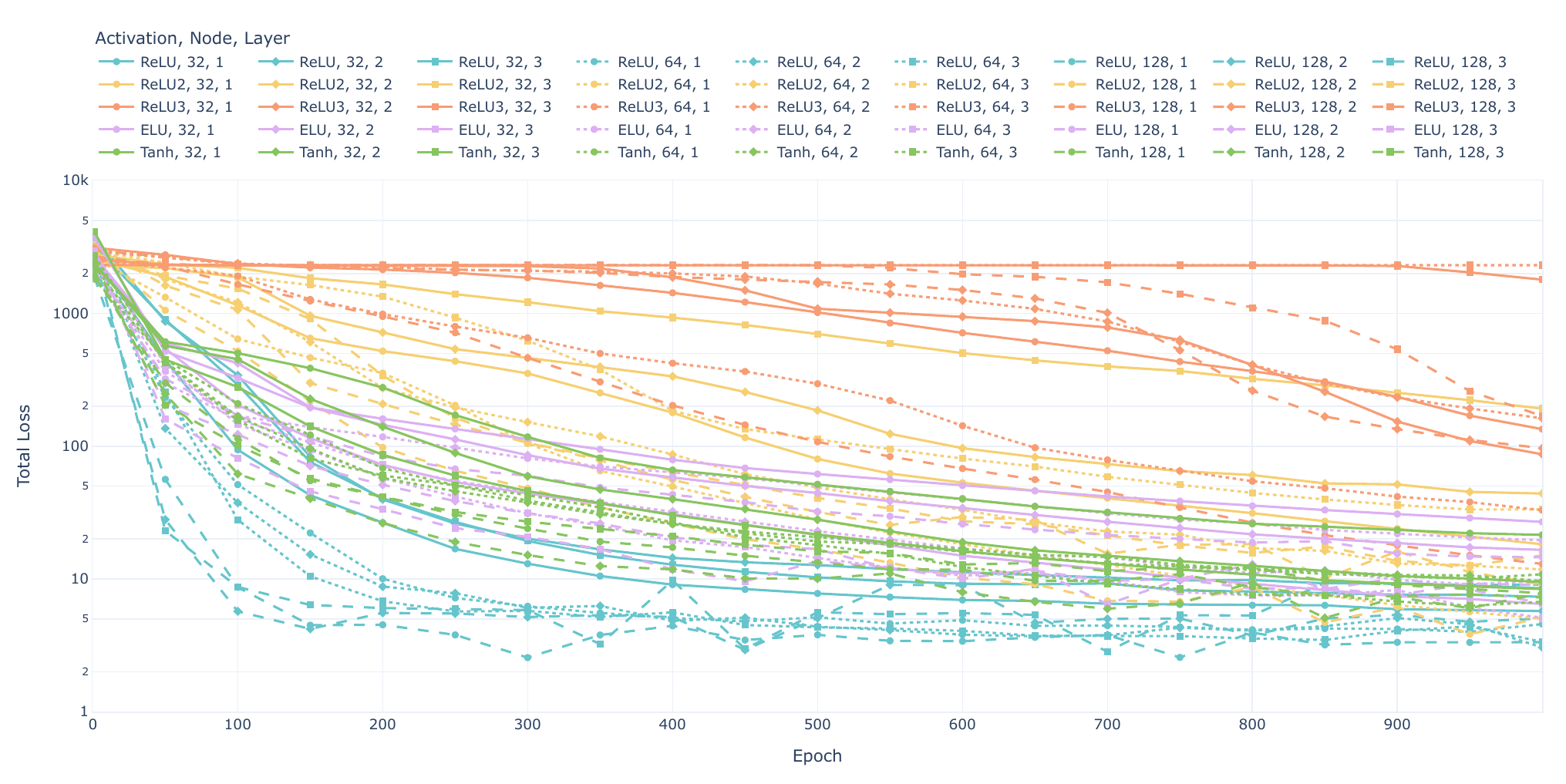}
    \caption{Epoch-wise total losses defined in Equation~\eqref{eq:L_epoch} of the $45$ neural representations specified in Equation~\eqref{eq:nn_set}. Unit of loss is basis points (bps), $1 \text{bps} = 10^{-4}$. Color corresponds to activation, line type to width, and marker shape to depth.}
    \label{fig:epoch_loss}
\end{figure}

Recall that Figure~\ref{fig:epoch_loss} only indicates the in-sample price fitting capacity of models with $\kD_T$, it remains to examine the density loss $\cL_D$ to confirm if the risk-neutral density is learned correctly, and check losses evaluated with the validation set $\kD_V$ to understand the generalization ability of the models. These are illustrated by Figure~\ref{fig:loss_scatter} through four scatter plots of model performances on: (a) training set price loss versus validation set price loss, $\cL_P(\kD_T)$ versus $\cL_P(\kD_V)$, (b) training set density loss versus validation set density loss, $\cL_D(\kD_T)$ versus $\cL_D(\kD_V)$, (c) training set price losses versus training set density loss, $\cL_P(\kD_T)$ versus $\cL_D(\kD_T)$, and (d) validation set price loss versus validation set density loss, $\cL_P(\kD_V)$ versus $\cL_D(\kD_V)$. Points with excessively large losses are trimmed out. We henceforth examine the activations more closely.

\begin{figure}[ht]
     \centering
        \includegraphics[width=\textwidth]{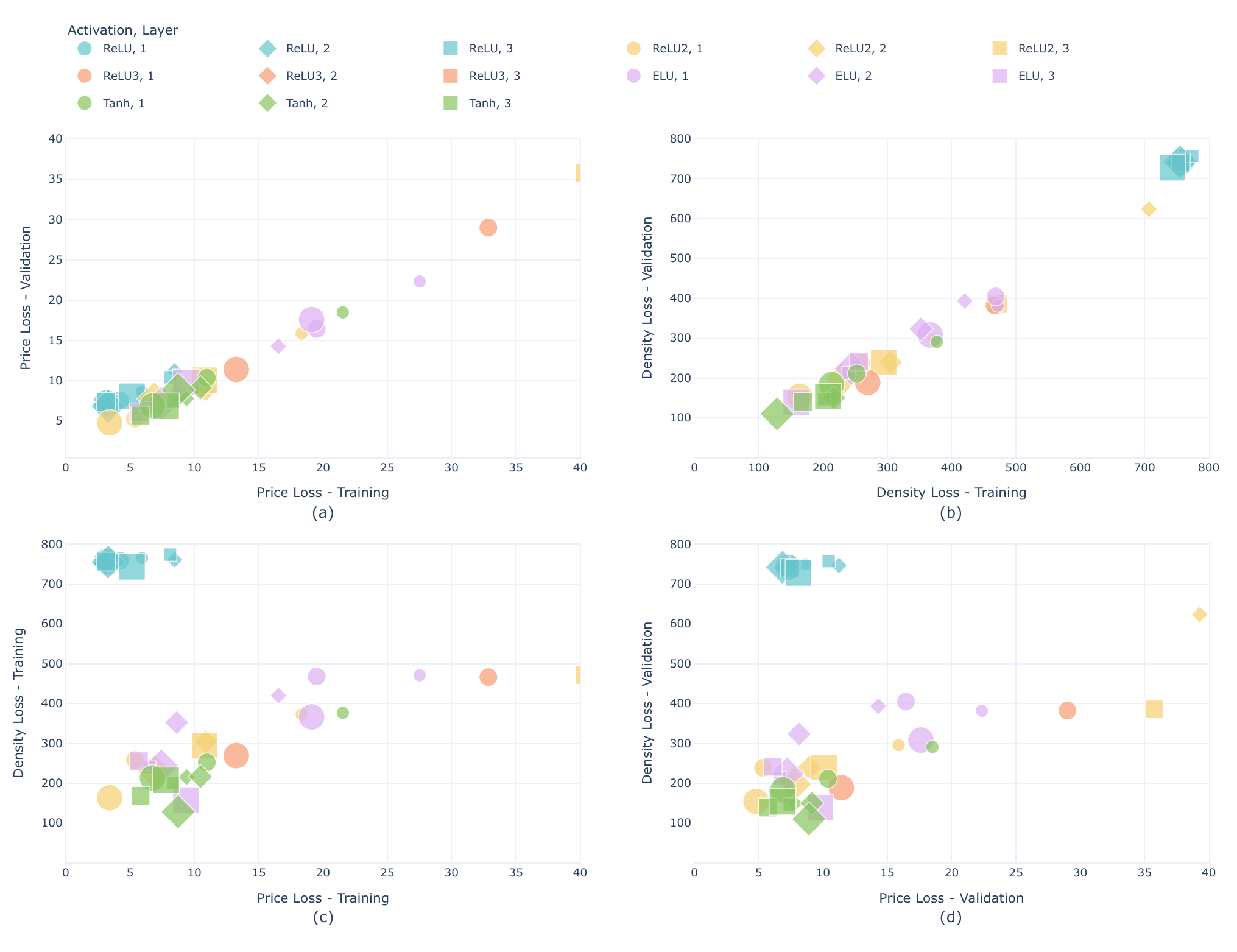}
    \caption{Scatter plots on price loss and density loss of training set and validation set. Color corresponds to activation, marker shape to depth, and marker size to width.}
    \label{fig:loss_scatter}
\end{figure}

\textbf{ReLU type neural representations.} Although $\cN_{\text{ReLU}}$ in Figure~\ref{fig:epoch_loss} seem to perform well with final epoch-wise price loss less than $15$ bps, Figure~\ref{fig:loss_scatter} (a) shows that $\cN_{\text{ReLU}}$ systematically deteriorate in validation price loss. For instance, $\cN_{\text{ReLU}}^{64, 3}$ has $\cL_P(\kD_T) \approx 3$ bps but $\cL_P(\kD_V) \approx 7$ bps, the latter being over twice that of the former. Such noticeable validation deterioration is not found on other activation types. Despite the suggested inferior generalization ability, $\cN_{\text{ReLU}}$ still perform in the top tier if we only consider the price loss. However, density losses in Figure~\ref{fig:loss_scatter} (b)-(d) further expose that $\cN_{\text{ReLU}}$ dramatically diverge from the correct risk-neutral density with $\cL_D(\{\kD_T, \kD_V\}) \approx 750$ bps. Figure~\ref{fig:relus} (a) and (b) clearly illustrates that $\cN_{\text{ReLU}}^{64, 3}$ mismatches the implied density in spite of the precise fitting on implied volatility. This phenomenon can be attributed to the piecewise-linear nature of the ReLU activation, whose indicator first derivative and zero second derivative oversimplify the neural derivatives in Equations~\eqref{eq:d_omega_hat} and~\eqref{eq:dd_omega_hat}, limiting their ability to approximate the correctors in Equation~\eqref{eq:zeta_xi}.

\begin{figure}[ht]
     \centering
     \begin{subfigure}[b]{0.4\textwidth}
         \centering
         \includegraphics[width=\textwidth, trim=1cm 0cm 1cm 2cm, clip=true]{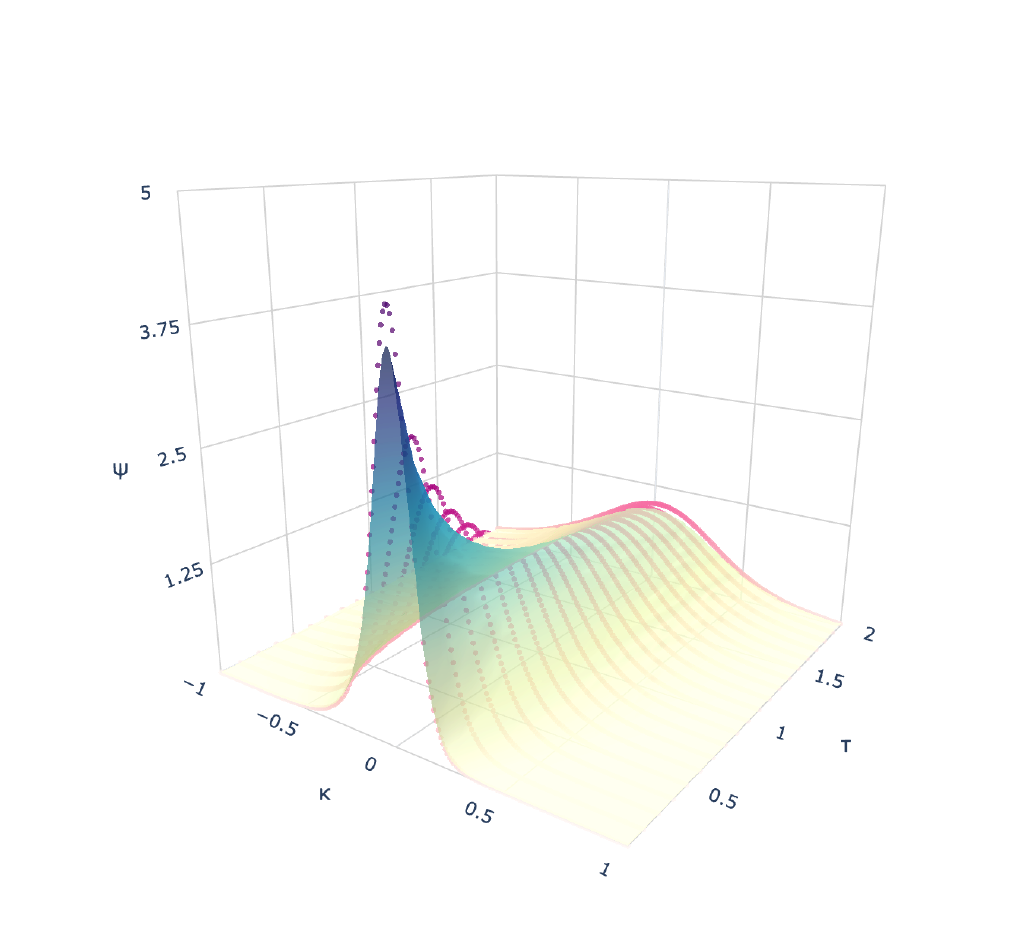}
         \caption{Implied density of $\cN_{\text{ReLU}}^{64, 3}$}
     \end{subfigure}
     \begin{subfigure}[b]{0.4\textwidth}
         \centering
         \includegraphics[width=\textwidth, trim=1cm 0cm 1cm 2cm, clip=true]{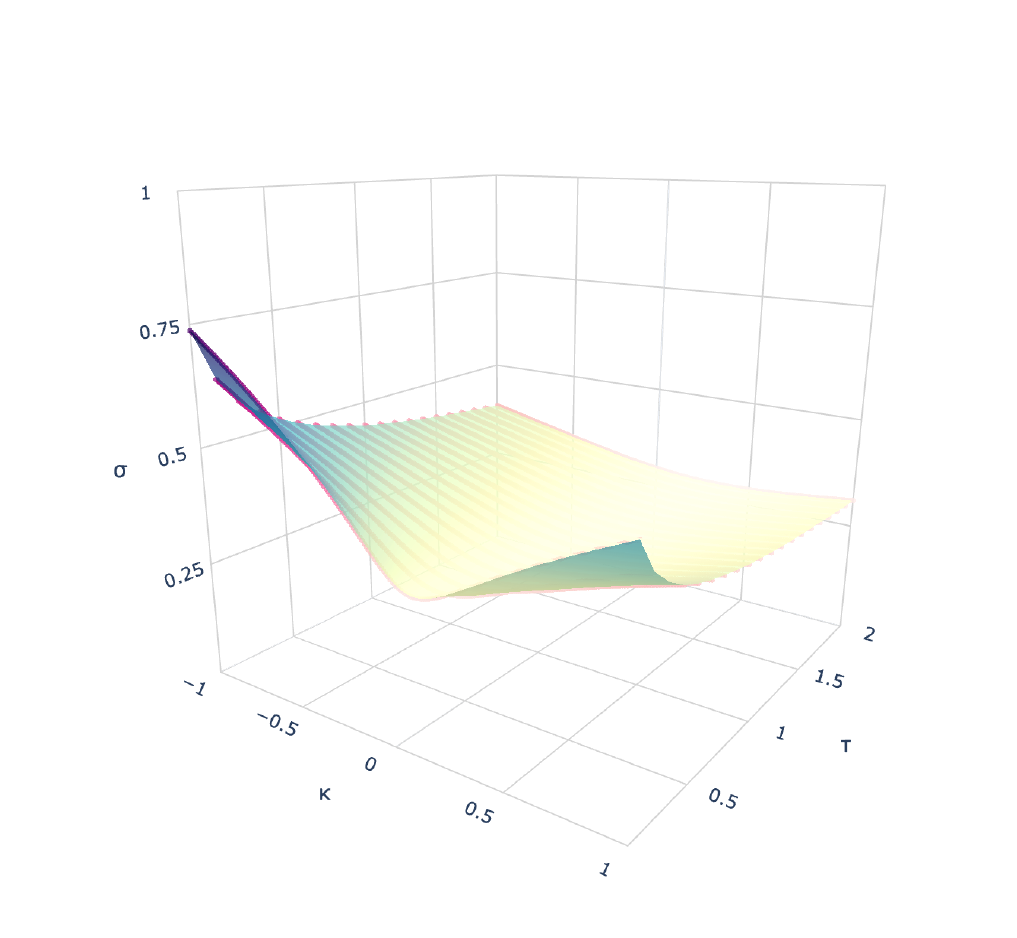}
         \caption{Implied volatility of $\cN_{\text{ReLU}}^{64, 3}$}
     \end{subfigure}
     \\
      \begin{subfigure}[b]{0.4\textwidth}
         \centering
         \includegraphics[width=\textwidth, trim=1cm 0cm 1cm 2cm, clip=true]{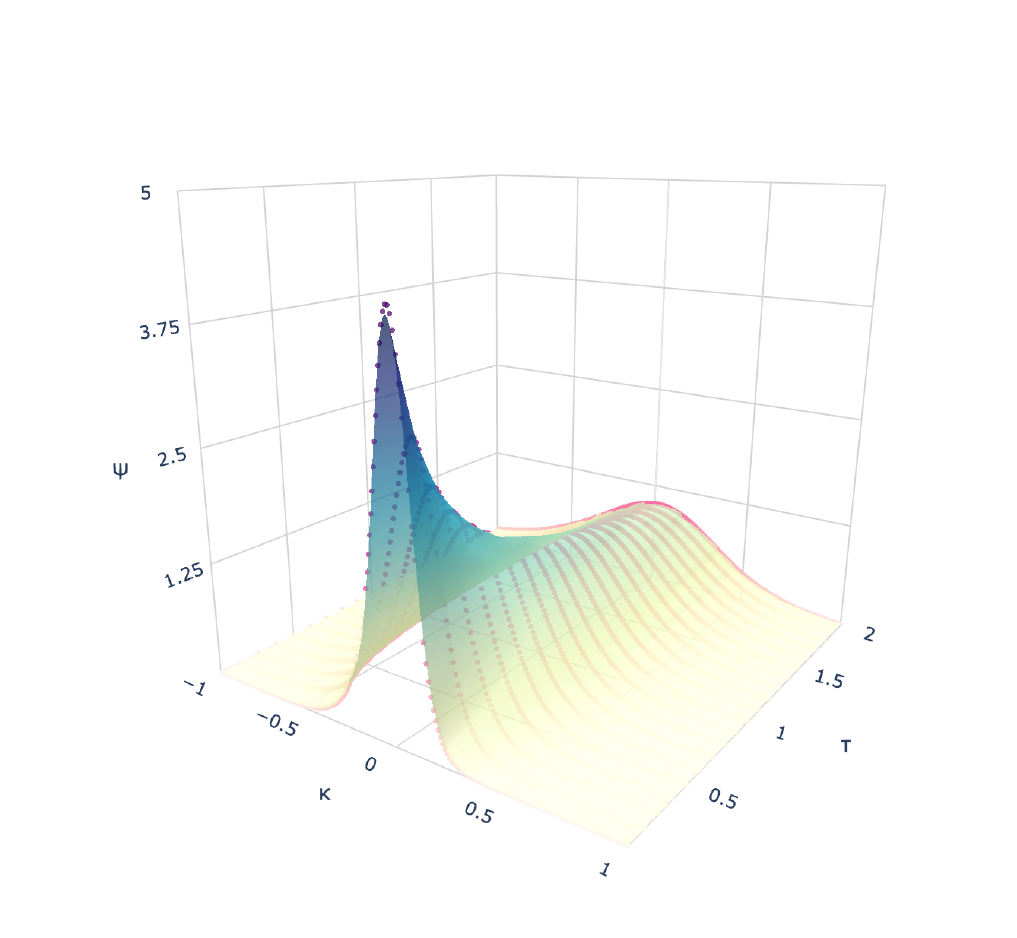}
         \caption{Implied density of $\cN_{\text{ReLU2}}^{128, 1}$}
     \end{subfigure}
      \begin{subfigure}[b]{0.4\textwidth}
         \centering
         \includegraphics[width=\textwidth, trim=1cm 0cm 1cm 2cm, clip=true]{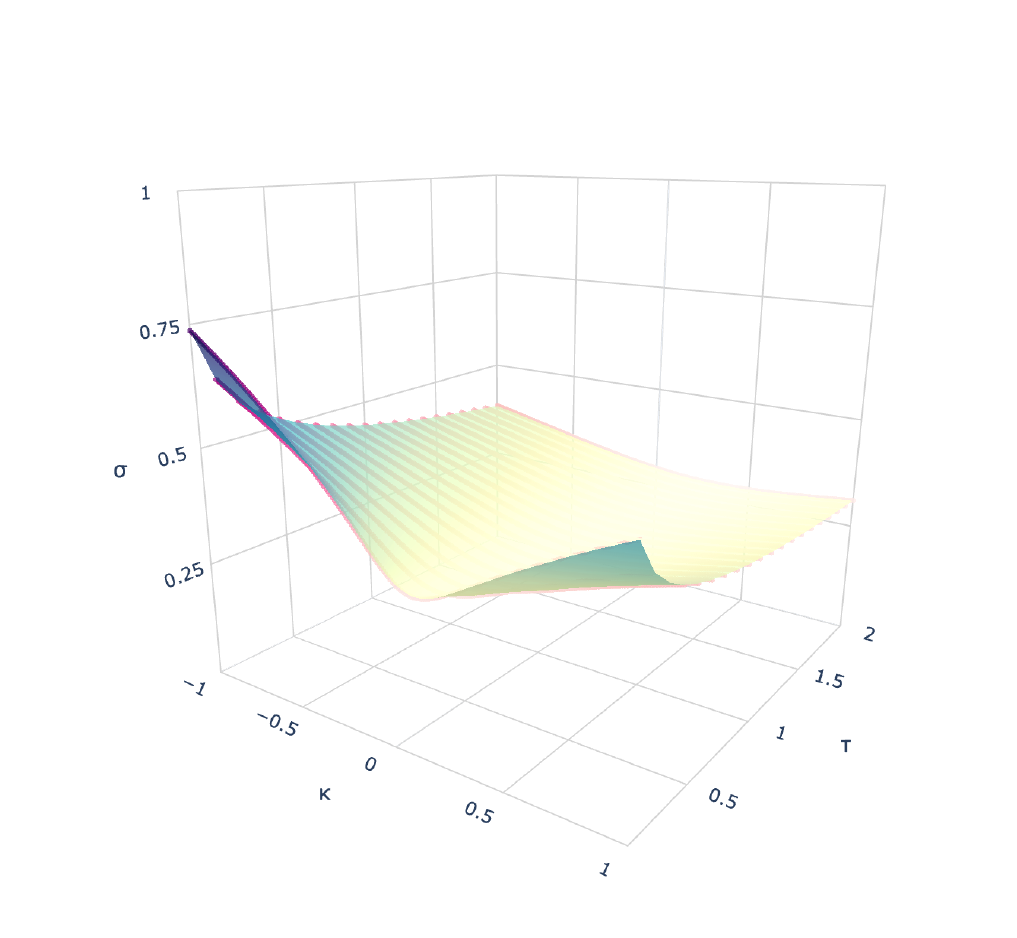}
         \caption{Implied volatility of $\cN_{\text{ReLU2}}^{128, 1}$}
     \end{subfigure}
    \caption{Surfaces of $\cN_{\text{ReLU}}^{64, 3}$ and $\cN_{\text{ReLU2}}^{128, 1}$. Red scatters are data points of the training set $\kD_T$.}
    \label{fig:relus}
\end{figure}

\textbf{ReLU2 type neural representations.} $\cN_{\text{ReLU2}}$ is worth noting as their performance spreads widely from top performer to bottom performer depending on the neural network structure. Their performance are consistent in both price loss and density loss, ordered by $\cN_{\text{ReLU2}}^{32, 3} \prec \cN_{\text{ReLU2}}^{32, 2} \prec \cN_{\text{ReLU2}}^{64, 3} \prec \cN_{\text{ReLU2}}^{32, 1} \prec \cN_{\text{ReLU2}}^{64, 2} \simeq \cN_{\text{ReLU2}}^{128, 3} \prec \cN_{\text{ReLU2}}^{64, 1} \simeq \cN_{\text{ReLU2}}^{128, 2} \prec \cN_{\text{ReLU2}}^{128, 1}$. The two close matching pairs, $\cN_{\text{ReLU2}}^{64, 2} \simeq \cN_{\text{ReLU2}}^{128, 3}$ and $\cN_{\text{ReLU2}}^{64, 1} \simeq \cN_{\text{ReLU2}}^{128, 2}$, well illustrate the trade-off between increased expressivity and more complicated derivative in Equation~\eqref{eq:d_omega_hat} and \eqref{eq:dd_omega_hat} for deeper neural representation. Remarkably, $\cN_{\text{ReLU2}}^{128, 1}$ is the only single hidden layer representation that have top tier performance, with $\cL_P(\{\kD_T, \kD_V\}) \approx 5$ bps and $\cL_D(\{\kD_T, \kD_V\}) \approx 160$ bps. Figure~\ref{fig:relus} (c) and (d) shows the precise fit of $\cN_{\text{ReLU2}}^{128, 1}$ on both implied volatility and implied density.

\textbf{ReLU3 type neural representations.} As mentioned above, $\cN_{\text{ReLU3}}$ generally fail except weak performer $\cN_{\text{ReLU3}}^{\{64, 128\}, 1}$. We notice that roughly $\cN_{\text{ReLU3}}^{64, 1} \simeq \cN_{\text{ReLU2}}^{64, 3}$ and $\cN_{\text{ReLU3}}^{128, 1} \simeq \cN_{\text{ReLU2}}^{128, 3}$, which agrees with \cite{cabanilla2024neural} that neural networks with ReLU powers need less depth. But in our context, $\cN_{\text{ReLU3}}$ could not achieve better results than the shallow $\cN_{\text{ReLU2}}^{128, 1}$, due to nonlinearity in both constraints and neural derivatives.

\textbf{ELU type neural representations.}
$\cN_{\text{ELU}}$ present a more obvious pattern in contrast to the mixed order of $\cN_{\text{ReLU2}}$ by depth and width. $\cN_{\text{ELU}}^{32, 1} \prec \cN_{\text{ELU}}^{64, 1} \simeq \cN_{\text{ELU}}^{128, 1} \prec \cN_{\text{ELU}}^{32, 2} \prec \cN_{\text{ELU}}^{64, 2} \prec \cN_{\text{ELU}}^{128, 2}$. This pattern breaks with deeper structure, as $\cN_{\text{ELU}}^{64, 3} \prec \cN_{\text{ELU}}^{32, 3} \prec \cN_{\text{ELU}}^{128, 2} \prec \cN_{\text{ELU}}^{128, 3}$ in price loss.

\textbf{Tanh type neural representations.}
We do not find a general pattern for $\cN_{\text{Tanh}}$. $\cN_{\text{Tanh}}^{\{32, 64, 128\}, 1}$ show improving results with width, while $\cN_{\text{Tanh}}^{\{32, 64, 128\}, \{2, 3\}}$ do not. 
Price loss is reduced with depth $\cN_{\text{Tanh}}^{64, 3} \prec \cN_{\text{Tanh}}^{128, 3} \prec \cN_{\text{Tanh}}^{32, 3} \prec \cN_{\text{Tanh}}^{128, 2} \simeq \cN_{\text{Tanh}}^{32, 2} \prec \cN_{\text{Tanh}}^{64, 2}$, but $\cN_{\text{Tanh}}^{128,2}$ has the lowest density loss. $\cN_{\text{Tanh}}$ in general performs better than $\cN_{\text{ELU}}$.

\begin{figure}[ht]
     \centering
        \includegraphics[width=\textwidth]{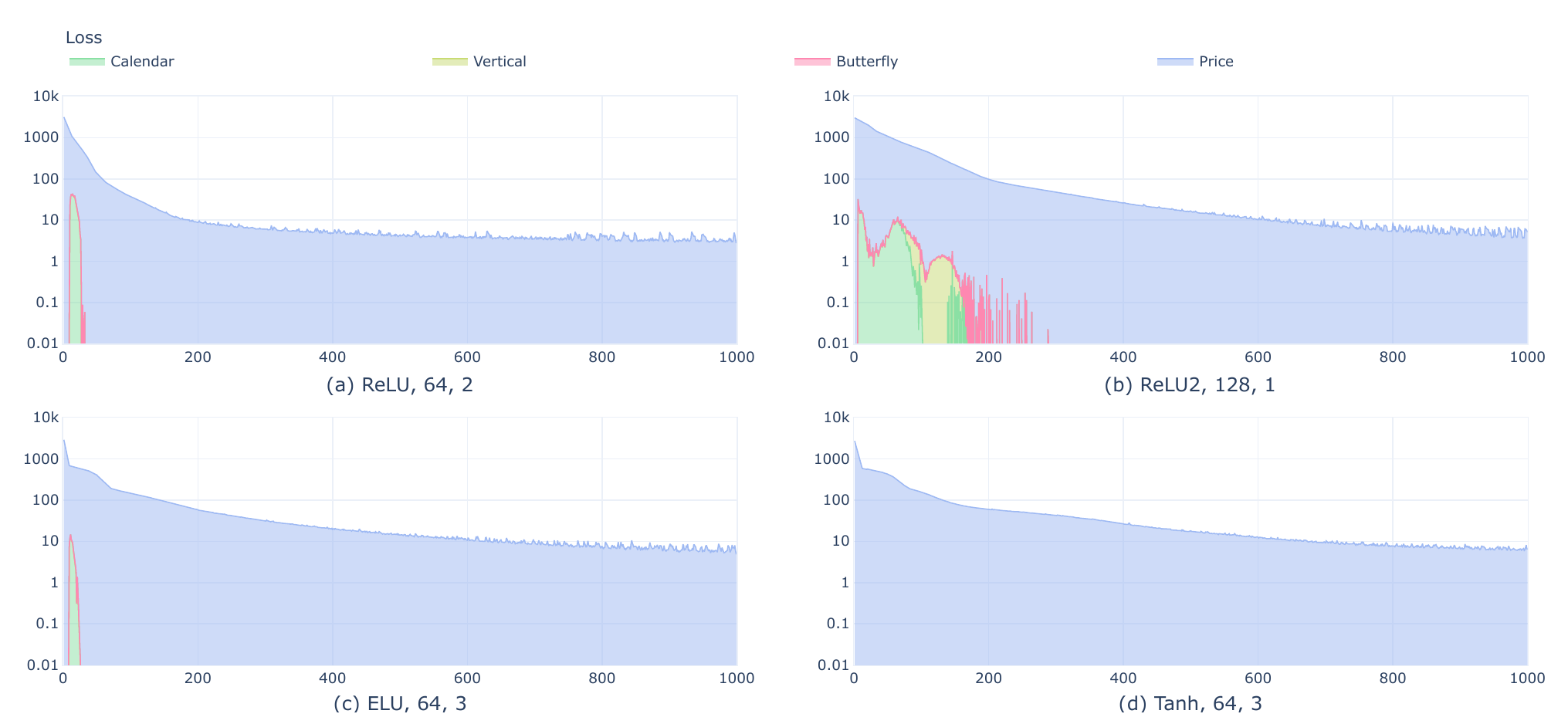}
    \caption{Itemized learning curves of top performers.}
    \label{fig:loss_type_best}
\end{figure}

With the above analysis, we identify some neural representations of interests. Some are the top performers within their activation types, and some show curiously reduced performance compared to their shallower or narrower counterparts. We thus examine their itemized learning curves Figure~\ref{fig:loss_type_best} reports within group top performers, $\cL_{\{P, C, V, B\}}^e$ of $\cN_{\text{ReLU}}^{64, 2}$, $\cN_{\text{ReLU2}}^{128, 1}$, $\cN_{\text{ELU}}^{64, 3}$, and $\cN_{\text{Tanh}}^{128, 3}$. All of them easily become free of arbitrage in a few epochs. Figure~\ref{fig:loss_type_best} (b) shows that $\cN_{\text{ReLU}}^{128, 1}$ requires slightly more iterations to rule out arbitrage.

\begin{figure}[ht]
     \centering
        \includegraphics[width=\textwidth]{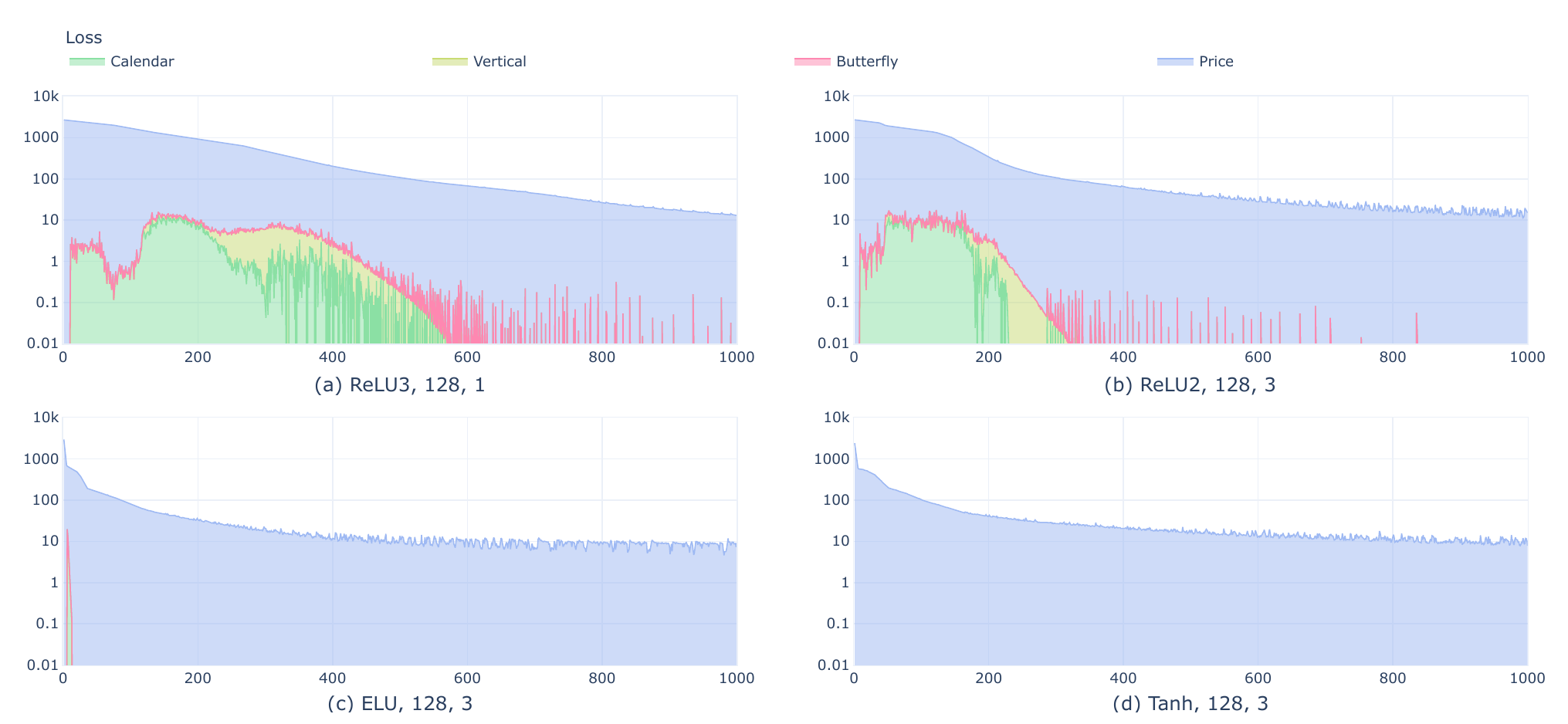}
    \caption{Itemized learning curves of noteworthy neural representations.}
    \label{fig:loss_type_curious}
\end{figure}

Figure~\ref{fig:loss_type_curious} collects $\cN_{\text{ReLU3}}^{128, 1}$, $\cN_{\text{ReLU2}}^{128, 3}$, $\cN_{\text{ELU}}^{128, 3}$, and $\cN_{\text{Tanh}}^{128, 3}$. We have pointed out earlier that $\cN_{\text{ReLU3}}^{128, 1} \simeq \cN_{\text{ReLU2}}^{128, 3}$, and Figure~\ref{fig:loss_type_curious} (a) and (b) indicates that they seem to suffer the similar difficulty in excluding arbitrage. Price loss curves in Figure~\ref{fig:loss_type_curious} (c) and (d) are noisier and reach their plateau earlier than those in Figure~\ref{fig:loss_type_best} (c) and (d), which leads to the inferior performance of $\cN_{\text{ELU}}^{128, 3}$ and $\cN_{\text{Tanh}}^{128, 3}$ to $\cN_{\text{ELU}}^{64, 3}$ and $\cN_{\text{Tanh}}^{64, 3}$.

\section{Discussion}

This study revisits the representation of option implied information from three complementary perspectives. Section~\ref{sec:option} simplified classical option pricing by reassembling familiar variable transformations, leading to an inherent standardization of data. Section~\ref{sec:volatility} reframed implied volatility as a pointwise corrector that morphs the Black-Scholes quasi-density into the market-implied risk-neutral density, thereby establishing an explicit density-volatility parity and clarifying the theoretical basis for arbitrage-free constraints. Section~\ref{sec:neural} demonstrated that this parity can be operationalized within a neural-network framework and that even shallow feedforward structures are capable of encoding both realistic implied volatility and implied density under arbitrage conditions.

Building a neural representation of option implied information is, however, a delicate task in two respects. First, conventional goodness of fit to implied volatility does not ensure a correct implied density, since the latter remains unobservable in practice. The $\cN_{\text{ReLU}}$ case exemplifies this modeling risk: excellent price fitting may coexist with an incorrect implied density, potentially leading to mispricing of exotic derivatives. Second, there is no universal rule of thumb for network design. Because of the nonlinearity in the arbitrage constraints and neural derivatives, the effects of depth, width, and activation smoothness vary across architectures. Upscaling networks does not necessarily improve performance, as seen in patterns such as $\cN_{\{\text{ELU}, \text{Tanh}\}}^{128,3} \prec \cN_{\{\text{ELU}, \text{Tanh}\}}^{64,3}$ and $\cN_{\text{ReLU2}}^{\{32,64,128\},3} \prec \cN_{\text{ReLU2}}^{\{32,64,128\},2} \prec \cN_{\text{ReLU2}}^{\{32,64,128\},1}$. Increasing activation smoothness does not always flatten the effective representation, and higher-order ReLU variants do not outperform the shallower quadratic case. Nonetheless, the performance of $\cN_{\text{ReLU2}}^{\{64, 128\},1}$ confirms that a simple, single-hidden-layer network with appropriate activations can provide an accurate and arbitrage-free representation of option implied information.

Overall, the findings demonstrate that shallow neural networks offer a concise yet expressive basis for modeling option implied information. The proposed representation flexibly captures a potentially wide range of realistic implied-volatility and implied-density surfaces at a given point in time and naturally lends itself to future investigation of their temporal evolution.

%

\setcitestyle{numbers}
\bibliographystyle{chicago}
\bibliography{ref}

@book{balakrishnan1991handbook,
  title={Handbook of the Logistic Distribution},
  author={Balakrishnan, Narayanaswamy},
  year={1991},
  publisher={CRC Press}
}

@article{carr2005note,
  title={A note on sufficient conditions for no arbitrage},
  author={Carr, Peter and Madan, Dilip B},
  journal={Finance Research Letters},
  volume={2},
  number={3},
  pages={125--130},
  year={2005},
  publisher={Elsevier}
}

@article{cox1976valuation,
  title={The valuation of options for alternative stochastic processes},
  author={Cox, John C and Ross, Stephen A},
  journal={Journal of Financial Economics},
  volume={3},
  number={1-2},
  pages={145--166},
  year={1976},
  publisher={Elsevier}
}

@article{merton1976option,
  title={Option pricing when underlying stock returns are discontinuous},
  author={Merton, Robert C},
  journal={Journal of Financial Economics},
  volume={3},
  number={1-2},
  pages={125--144},
  year={1976},
  publisher={Elsevier}
}

@article{bates1996jumps,
  title={Jumps and stochastic volatility: Exchange rate processes implicit in Deutsche Mark options},
  author={Bates, David S},
  journal={The Review of Financial Studies},
  volume={9},
  number={1},
  pages={69--107},
  year={1996},
  publisher={Oxford University Press}
}

@article{heston1993closed,
  title={A closed-form solution for options with stochastic volatility with applications to bond and currency options},
  author={Heston, Steven L},
  journal={The Review of Financial Studies},
  volume={6},
  number={2},
  pages={327--343},
  year={1993},
  publisher={Oxford University Press}
}

@article{dupire1994pricing,
  title={Pricing with a smile},
  author={Dupire, Bruno},
  journal={Risk},
  volume={7},
  number={1},
  pages={18--20},
  year={1994}
}

@article{black1973pricing,
  title={The pricing of options and corporate liabilities},
  author={Black, Fischer and Scholes, Myron},
  journal={Journal of Political Economy},
  volume={81},
  number={3},
  pages={637--654},
  year={1973},
  publisher={The University of Chicago Press}
}

@article{choi2022black,
  title={A Black--Scholes user's guide to the Bachelier model},
  author={Choi, Jaehyuk and Kwak, Minsuk and Tee, Chyng Wen and Wang, Yumeng},
  journal={Journal of Futures Markets},
  volume={42},
  number={5},
  pages={959--980},
  year={2022},
  publisher={Wiley Online Library}
}

@article{kingma2014adam,
  title={Adam: A method for stochastic optimization},
  author={Kingma, Diederik P and Ba, Jimmy},
  journal={arXiv preprint arXiv:1412.6980},
  year={2014}
}

@article{azzone2023fast,
  title={A fast Monte Carlo scheme for additive processes and option pricing},
  author={Azzone, Michele and Baviera, Roberto},
  journal={Computational Management Science},
  volume={20},
  number={1},
  pages={31},
  year={2023},
  publisher={Springer}
}

@article{carr2021additive,
  title={Additive logistic processes in option pricing},
  author={Carr, Peter and Torricelli, Lorenzo},
  journal={Finance and Stochastics},
  volume={25},
  number={4},
  pages={689--724},
  year={2021},
  publisher={Springer}
}

@article{ackerer2020deep,
  title={Deep smoothing of the implied volatility surface},
  author={Ackerer, Damien and Tagasovska, Natasa and Vatter, Thibault},
  journal={Advances in Neural Information Processing Systems},
  volume={33},
  pages={11552--11563},
  year={2020}
}

@article{hornik1989multilayer,
  title={Multilayer feedforward networks are universal approximators},
  author={Hornik, Kurt and Stinchcombe, Maxwell and White, Halbert},
  journal={Neural Networks},
  volume={2},
  number={5},
  pages={359--366},
  year={1989},
  publisher={Elsevier}
}

@article{hornik1990universal,
  title={Universal approximation of an unknown mapping and its derivatives using multilayer feedforward networks},
  author={Hornik, Kurt and Stinchcombe, Maxwell and White, Halbert},
  journal={Neural Networks},
  volume={3},
  number={5},
  pages={551--560},
  year={1990},
  publisher={Elsevier}
}

@article{hornik1991approximation,
  title={Approximation capabilities of multilayer feedforward networks},
  author={Hornik, Kurt},
  journal={Neural Networks},
  volume={4},
  number={2},
  pages={251--257},
  year={1991},
  publisher={Elsevier}
}

@inproceedings{eldan2016power,
  title={The power of depth for feedforward neural networks},
  author={Eldan, Ronen and Shamir, Ohad},
  booktitle={Conference on Learning Theory},
  pages={907--940},
  year={2016}
}

@article{cabanilla2024neural,
  title={Neural networks with ReLU powers need less depth},
  author={Cabanilla, Kurt Izak M and Mohammad, Rhudaina Z and Lope, Jose Ernie C},
  journal={Neural Networks},
  volume={172},
  pages={106073},
  year={2024},
  publisher={Elsevier}
}

@inproceedings{he2015delving,
  title={Delving deep into rectifiers: Surpassing human-level performance on imagenet classification},
  author={He, Kaiming and Zhang, Xiangyu and Ren, Shaoqing and Sun, Jian},
  booktitle={Proceedings of the IEEE International Conference on Computer Vision},
  pages={1026--1034},
  year={2015}
}

@article{wiedemann2024operator,
  title={Operator deep smoothing for implied volatility},
  author={Wiedemann, Ruben and Jacquier, Antoine and Gonon, Lukas},
  journal={arXiv preprint arXiv:2406.11520},
  year={2024}
}

@article{breeden1978prices,
  title={Prices of state-contingent claims implicit in option prices},
  author={Breeden, Douglas T and Litzenberger, Robert H},
  journal={Journal of Business},
  volume={51},
  number={4},
  pages={621--651},
  year={1978},
  publisher={JSTOR}
}

@article{durrleman2010implied,
  title={Implied volatility: Market models},
  author={Durrleman, Valdo},
  journal={Encyclopedia of Quantitative Finance},
  year={2010},
  publisher={Wiley Online Library}
}

@article{benko2007extracting,
  title={On extracting information implied in options},
  author={Benko, Michal and Fengler, Matthias and H{\"a}rdle, Wolfgang and Kopa, Milos},
  journal={Computational Statistics},
  volume={22},
  number={4},
  pages={543--553},
  year={2007},
  publisher={Springer}
}

@article{roper2010arbitrage,
  title={Arbitrage free implied volatility surfaces},
  author={Roper, Michael},
  journal={Preprint},
  year={2010}
}

@phdthesis{roper2009implied,
  title={Implied volatility: General properties and asymptotics},
  author={Roper, Michael},
  year={2009},
  school={UNSW Sydney}
}

@article{tavin2012implied,
  title={Implied distribution as a function of the volatility smile},
  author={Tavin, Bertrand},
  journal={Bankers Markets and Investors},
  number={119},
  pages={31--42},
  year={2012}
}

@article{brunner2003arbitrage,
  title={Arbitrage-free estimation of the risk-neutral density from the implied volatility smile},
  author={Brunner, Bernhard and Hafner, Reinhold},
  journal={Journal of Computational Finance},
  year={2003},
  volume={7},
  number={1},
  pages={75--106}
}

@article{jackwerth2000recovering,
  title={Recovering risk aversion from option prices and realized returns},
  author={Jackwerth, Jens Carsten},
  journal={The Review of Financial Studies},
  volume={13},
  number={2},
  pages={433--451},
  year={2000},
  publisher={Oxford University Press}
}

@article{figlewski2018risk,
  title={Risk-neutral densities: A review},
  author={Figlewski, Stephen},
  journal={Annual Review of Financial Economics},
  volume={10},
  number={1},
  pages={329--359},
  year={2018},
  publisher={Annual Reviews}
}

@article{bergomi2012stochastic,
  title={Stochastic volatility's orderly smiles},
  author={Bergomi, Lorenzo and Guyon, Julien},
  journal={Risk},
  volume={25},
  number={5},
  pages={60},
  year={2012},
  publisher={Incisive Media Limited}
}

@article{bayer2016pricing,
  title={Pricing under rough volatility},
  author={Bayer, Christian and Friz, Peter and Gatheral, Jim},
  journal={Quantitative Finance},
  volume={16},
  number={6},
  pages={887--904},
  year={2016},
  publisher={Taylor \& Francis}
}

@article{jaber2022quintic,
  title={The quintic Ornstein-Uhlenbeck volatility model that jointly calibrates SPX \& VIX smiles},
  author={Jaber, Eduardo Abi and Illand, Camille and Li, Shaun Xiaoyuan},
  journal={arXiv preprint arXiv:2212.10917},
  year={2022}
}

@article{gatheral2004parsimonious,
  title={A parsimonious arbitrage-free implied volatility parameterization with application to the valuation of volatility derivatives},
  author={Gatheral, Jim},
  journal={Presentation at Global Derivatives \& Risk Management, Madrid},
  year={2004}
}

@inproceedings{avellaneda2005sabr,
  title={From SABR to geodesics},
  author={Avellaneda, Marco},
  booktitle={Conference Presentation at Courant Institute; New York: New York University},
  year={2005}
}

@book{gatheral2006volatility,
  title={The volatility surface: a practitioner's guide},
  author={Gatheral, Jim},
  year={2006},
  publisher={Wiley}
}

@article{gatheral2014arbitrage,
  title={Arbitrage-free SVI volatility surfaces},
  author={Gatheral, Jim and Jacquier, Antoine},
  journal={Quantitative Finance},
  volume={14},
  number={1},
  pages={59--71},
  year={2014},
  publisher={Taylor \& Francis}
}

@article{fengler2009arbitrage,
  title={Arbitrage-free smoothing of the implied volatility surface},
  author={Fengler, Matthias R},
  journal={Quantitative Finance},
  volume={9},
  number={4},
  pages={417--428},
  year={2009},
  publisher={Taylor \& Francis}
}

@article{andreasen2011volatility,
  title={Volatility interpolation},
  author={Andreasen, Jesper and Huge, Brian},
  journal={Risk},
  volume={24},
  number={3},
  pages={76},
  year={2011},
  publisher={Incisive Media Limited}
}

@article{glaser2012arbitrage,
  title={Arbitrage-free approximation of call price surfaces and input data risk},
  author={Glaser, Judith and Heider, Pascal},
  journal={Quantitative Finance},
  volume={12},
  number={1},
  pages={61--73},
  year={2012},
  publisher={Taylor \& Francis}
}

@article{corbetta2019robust,
  title={Robust calibration and arbitrage-free interpolation of SSVI slices},
  author={Corbetta, Jacopo and Cohort, Pierre and Laachir, Ismail and Martini, Claude},
  journal={Decisions in Economics and Finance},
  volume={42},
  number={2},
  pages={665--677},
  year={2019},
  publisher={Springer}
}

@article{hendriks2019extended,
  title={The extended SSVI volatility surface},
  author={Hendriks, Sebas and Martini, Claude},
  journal={Journal of Computational Finance},
  volume={22},
  pages={25--39},
  year={2019}
}

@inproceedings{zheng2021incorporating,
  title={Incorporating prior financial domain knowledge into neural networks for implied volatility surface prediction},
  author={Zheng, Yu and Yang, Yongxin and Chen, Bowei},
  booktitle={Proceedings of the 27th ACM SIGKDD Conference on Knowledge Discovery \& Data Mining},
  pages={3968--3975},
  year={2021}
}

@article{mingone2022no,
  title={No arbitrage global parametrization for the eSSVI volatility surface},
  author={Mingone, Arianna},
  journal={Quantitative Finance},
  volume={22},
  number={12},
  pages={2205--2217},
  year={2022},
  publisher={Taylor \& Francis}
}

@inproceedings{yang2025hyperiv,
  title={HyperIV: Real-time implied volatility smoothing},
  author={Yang, Yongxin and Chen, Wenqi and Shu, Chao and Hospedales, Timothy},
  booktitle={The 42nd International Conference on Machine Learning},
  pages={1--15},
  year={2025}
}

@article{bergeron2022variational,
  title={Variational Autoencoders: A hands-Off approach to volatility},
  author={Bergeron, Maxime and Fung, Nicholas and Hull, John and Poulos, Zissis and Veneris, Andreas},
  journal={The Journal of Financial Data Science},
  volume={4},
  number={2},
  pages={125--138},
  year={2022}
}

@article{vuletic2024volgan,
  title={VolGAN: A generative model for arbitrage-free implied volatility surfaces},
  author={Vuleti{\'c}, Milena and Cont, Rama},
  journal={Applied Mathematical Finance},
  volume={31},
  number={4},
  pages={203--238},
  year={2024},
  publisher={Taylor \& Francis}
}

@article{azzone2025explicit,
  title={Explicit option pricing with additive processes},
  author={Azzone, Michele and Torricelli, Lorenzo},
  journal={SIAM Journal on Financial Mathematics},
  volume={16},
  number={3},
  pages={747--802},
  year={2025},
  publisher={SIAM}
}

@inproceedings{lin2024neural,
  title={Neural Term Structure of Additive Process for Option Pricing},
  author={Lin, Jimin and Liu, Guixin},
  booktitle={Proceedings of the 5th ACM International Conference on AI in Finance},
  pages={695--702},
  year={2024}
}

@article{dotsis2025option,
  title={Option pricing before Black, Scholes and Merton: A review and assessment of the historical evidence},
  author={Dotsis, George},
  journal={SSRN 5355859},
  year={2025}
}

\appendix
\section{Appendix}
\begin{proof}[Proof of Proposition~\ref{prop:psi_omega}]
This follows by straightforward calculus. For compactness we omit $(\tau, \kappa)$. First calculate the total vega by
\begin{align}
\pd_\omega p_{BS}^\omega 
    &= e^\kappa \frac{\pd \Phi(z_+^\omega)}{\pd z_+^\omega}\frac{\pd z_+^\omega}{\pd \omega} - \frac{\pd \Phi(z_-^\omega)}{\pd z_-^\omega}\frac{\pd z_-^\omega}{\pd \omega} \\
    &= e^\kappa \phi(z_+^\omega) \frac{\pd z_+^\omega}{\pd \omega} -\phi(z_-^\omega)\frac{\pd z_-^\omega}{\pd \omega} \\
    &= \phi(z_-^\omega) \frac{\pd (z_+^\omega - z_-^\omega)}{\pd \omega} \\
    &= \phi(z_-^\omega),
\end{align}
where $e^\kappa \phi(z_+^\omega) = \phi(z_-^\omega)$ is used in the third line. Thus the calendar spread is given by
\begin{align}
\pd_\tau p = \pd_\tau p_{BS}^\omega = \pd_\omega p_{BS}^\omega \pd_\tau \omega = \phi(z_-^\omega) \pd_\tau \omega.
\end{align}
The implied CDF is calculated by
\begin{align}
\Psi
    &= e^{-\kappa} \pd_\kappa p \\
    &= e^{-\kappa} \pd_\kappa p_{BS}^\omega \\
    &= e^{-\kappa} \left( e^\kappa \Phi(z_+^\omega) + \pd_\omega p_{BS}^\omega \pd_\kappa \omega\right) \\
    &= \Phi(z_+^\omega) + e^{-\kappa}\phi(z_-^\omega) \pd_\kappa \omega \\
    &= \Phi(z_+^\omega) + \phi(z_+^\omega) \pd_\kappa \omega.
\end{align}
The implied PDF is then computed as
\begin{align}
\psi
    &= \pd_\kappa \Psi \\
    &= \pd_z\Phi(z_+^\omega) \pd_\kappa z_+^\omega + \pd_z \phi(z_+^\omega) \pd_\kappa z_+^\omega \pd_\kappa \omega + \phi(z_+^\omega) \pd_{\kappa \kappa} \omega \\
    &= \phi(z_+^\omega) \pd_\kappa z_+^\omega - z_+^\omega \phi(z_+^\omega)\pd_\kappa z_+^\omega \pd_\kappa \omega + \phi(z_+^\omega) \pd_{\kappa \kappa} \omega\\
    &= \phi(z_+^\omega) \left( \pd_\kappa z_+^\omega (1 - z_+^\omega \pd_\kappa\omega) + \pd_{\kappa \kappa} \omega\right) \\
    &= \phi(z_+^\omega) \left( \frac{1}{\omega}(1 - z_-^\omega \pd_\kappa \omega) (1 - z_+^\omega \pd_\kappa\omega) + \pd_{\kappa \kappa} \omega\right) \\
    &= \frac{\phi(z_+^\omega)}{\omega} \left( 1 - (z_+^\omega  + z_-^\omega) \pd_\kappa \omega + z_-^\omega z_+^\omega (\pd_\kappa\omega)^2 + \omega \pd_{\kappa \kappa} \omega\right) \\
    &= \frac{\phi(z_+^\omega)}{\omega} \left( 1 - \frac{2 \kappa}{\omega} \pd_\kappa \omega + \left(\frac{\kappa^2}{\omega^2} - \frac{1}{4}\omega^2\right) (\pd_\kappa\omega)^2 + \omega \pd_{\kappa \kappa} \omega\right) \\
    &= \frac{\phi(z_+^\omega)}{\omega} \left( \left(1 - \frac{\kappa}{\omega} \pd_\kappa \omega \right)^2  - \frac{1}{4}(\omega^2 \pd_\kappa\omega)^2 + \omega \pd_{\kappa \kappa} \omega\right),
\end{align}
where $\pd_z\phi(z)= -z \phi(z)$ is used in the third line and $\pd_\kappa z_+^\omega = \frac{1}{\omega} - (\frac{\kappa}{\omega^2} - \frac{1}{2}) \pd_\kappa \omega = \frac{1}{\omega}(1 - z_-^\omega \pd_\kappa \omega)$ is used in the fifth line. Substitute $\psi_{BS}^\omega$, $\wtilde{\psi}_{BS}^\omega$, $\Psi_{BS}^\omega$ into the above equations and the result follows.

\end{proof}

\end{document}